\documentclass[11pt]{article}

\usepackage{amsfonts, caption, subcaption}
\usepackage{amsmath, amsthm}
\usepackage{graphicx}
\usepackage{hyperref}
\usepackage{amssymb}
\usepackage{amsopn}
\usepackage{comment}
\usepackage{color}
\usepackage{tikz}
\usetikzlibrary{positioning,chains,fit,shapes,calc}

\usepackage[a4paper]{geometry}
\newtheorem{theorem}{Theorem}%[section]
\newtheorem{lemma}[theorem]{Lemma}

\newtheorem{proposition}[theorem]{Proposition}

\newtheorem{corollary}[theorem]{Corollary}
\newtheorem{definition}[theorem]{Definition}
\newtheorem{example}[theorem]{Example}

\newtheorem{remark}[theorem]{Remark}

\newcommand{\E}{\mathbb{E}}

\newcommand{\Var}{\text{{\rm{Var}}}}

\newcommand{\B}{\{0,1\}}
\begin{document}
\title{Deterministic Randomness Extraction from Generalized and Distributed Santha-Vazirani Sources}
\author{Salman Beigi$^1$, Omid Etesami$^1$, and Amin Gohari$^{1,2}$
\\ \emph{\small $^1$School of Mathematics, Institute for Research in Fundamental Sciences (IPM), Tehran, Iran}\\
\emph{\small $^2$Department of Electrical Engineering, Sharif University of Technology, Tehran, Iran}}
\maketitle

\begin{abstract}
A Santha-Vazirani (SV) source is a sequence of random bits where the conditional distribution of each bit, given the previous bits, can be partially controlled by an adversary.  Santha and Vazirani show that deterministic randomness extraction from these sources is impossible.
In this paper, we study the generalization of SV sources for non-binary sequences. We show that unlike the binary case, deterministic randomness extraction in the generalized case is sometimes possible. We present a necessary condition and a sufficient condition for the possibility of deterministic randomness extraction. These two conditions coincide in ``non-degenerate" cases.

Next, we turn to a distributed setting.
In this setting the SV source consists of a random sequence of pairs $(a_1, b_1), (a_2, b_2), \ldots$ distributed between two parties, where the first party receives $a_i$'s and the second one receives $b_i$'s.
The goal of the two parties is to extract  common randomness without communication. Using the notion of \emph{maximal correlation}, we prove a necessary condition and a sufficient condition for the possibility of common randomness extraction from these sources. Based on these two conditions, the problem of common randomness extraction essentially reduces to the problem of randomness extraction from (non-distributed) SV sources. This result generalizes results of G\'acs and K\"orner, and Witsenhausen about common randomness extraction from i.i.d.\ sources to adversarial sources.

\end{abstract}

%%%%%%%%%%%%%%%%%%%%%%%%%%%%%%%%%%%%%%%%%%%%%%%%%%%%%%%%%%%%%%%%

\section{Introduction}

Randomized algorithms are simpler and more efficient than their deterministic counterparts in many applications.
In some settings such as communication complexity and distributed computing, it is even possible to prove unconditionally that allowing randomness improves
the efficiency of algorithms (see e.g.,~\cite{Yao79, Rabin, FischerLynch}).
However, access to sources of randomness (especially common randomness) may be limited, or the quality of randomness in the source may be far from perfect. Having such an imperfect source of randomness, one may be able to extract (almost) unbiased and independent random bits using \emph{randomness extractors}. A randomness extractor is a function applied to an imperfect source of randomness whose outcome is an almost perfect source of randomness.

The problem of randomness extraction from imperfect sources of randomness was perhaps first considered by Von Neumann \cite{vonNeumann}.
A later important work in this area
is \cite{SanthaVazirani} where Santha and Vazirani introduced the imperfect sources of randomness
now often called Santha-Vazirani (SV) sources.
These sources can easily be defined in terms of an adversary with two coins.
Consider an adversary who has two different coins, one of which is biased towards heads (e.g.,\ $\Pr(\text{heads}) = 2/3$)
and the other one is biased towards tails (e.g.,\ $\Pr(\text{heads}) = 1/3$).
The adversary, in each time step, chooses one of the two coins and tosses it.
Adversary's choice of coin may depend (probabilistically) on the previous outcomes of the tosses.
The sequence of random outcomes of these coin tosses is called a SV source.

Santha and Vazirani \cite{SanthaVazirani} show that randomness extraction from the above sources through a deterministic method is impossible. More precisely, they show that for every deterministic way of extracting one random bit, there is a strategy for the adversary
such that the extracted bit is biased, or more specifically,
the extracted bit is 0 with probability either $\ge 2/3$ or $\le 1/3$.
Subsequently, other proofs for this result have been found (see e.g.,~\cite{RVW, ACMPS}).  In Appendix \ref{appendixC}, we give a more refined version of this result, which provides a more detailed picture of the limits of what the adversary can achieve.

Despite this negative result, such imperfect sources of randomness are enough for many applications.
For example, as shown by Vazirani and Vazirani \cite{VV1,VV2}, randomized polynomial-time algorithms that use perfect random bits
can be simulated using SV sources. This fact can also be verified using the fact that the min-entropy of SV sources is linear in the size of the source (where min-entropy, in the context of extractors, was first introduced by \cite{CG}). Indeed, by the  later theory of randomness extraction (e.g.,~see~\cite{Zuckerman}), it is possible to efficiently extract polynomially many almost random bits
from such sources with high min-entropy if we are, in addition to the imperfect source, endowed with a perfectly random seed of logarithmic length. (In fact, for the special case of SV sources, a seed of constant length is enough~\cite[Problem 6.6]{Vadhan}). For the application of randomized polynomial-time algorithms where a logarithmic-length random seed is not available,
we can enumerate in polynomial time over all the logarithmic-length seeds;
for each choice of the seed, we apply the randomness extractor with the given seed and use its outcome bits (that are not truly random) in the algorithm.
Finally, we take a majority over the outputs of the algorithm for different choices of the seed.

Enumerating over all seeds may be inefficient for some applications,
or does not work at all, e.g., in interactive proofs and one-shot scenarios such as cryptography.
Therefore, it is natural to ask whether deterministic randomness extraction from imperfect sources of randomness is possible.
For most applications, it is also necessary to require that the extractor be explicit, i.e., extraction can be done efficiently (in polynomial time).
Previous to this work, explicit deterministic extractors had been constructed for many different classes of sources, including i.i.d.\ bits with unknown bias \cite{vonNeumann}, Markov chains \cite{Blum}, affine sources \cite{Bourgain,GR}, polynomial sources  \cite{DGW,Dvir}, and
sources consisting of independent blocks  \cite{Bourgain2}.

\vspace{.2in}
\noindent
{\bf Deterministic extractors for generalized SV sources.}
Although~\cite{SanthaVazirani} proves the
impossibility of deterministic randomness extraction from SV sources,
this impossibility is shown only for binary sources.
In this paper we show that if we consider a generalization of SV sources over \emph{non-binary} alphabets,
deterministic randomness extraction is indeed possible under certain conditions.

To generalize SV sources over non-binary alphabets,
we assume that the adversary, instead of coins, has some multi-faceted (say 6-sided) dice.
The numbers written on the faces of different dice are the same,
but each die may have a different probability for a given face value.
The adversary throws these dice $n$ times,
each time choosing a die to throw depending on the results of the previous throws. Again, the outcome is an imperfect source of randomness, for which we may ask whether deterministic randomness extraction is possible or not.

When the dice are non-degenerate, i.e., all faces of all dice have non-zero probability,
we give a necessary and sufficient condition for the existence of
a deterministic strategy for extracting one bit with arbitrarily small bias.
For example, when the dice are 6-sided,
the necessary and sufficient condition implies that we can deterministically extract an almost unbiased bit when the adversary has access to any arbitrary set of five non-degenerate dice, but randomness extraction is not possible in general when the adversary has access to six non-degenerate 6-sided dice.
More precisely, a set of non-degenerate dice leads to extractable generalized SV sources
if and only if
the convex hull of the set of probability distributions associated with the set of dice
does not have full dimension in the ``probability simplex".
We emphasize that when we prove the possibility of deterministic extraction,
we also provide an explicit extractor.

\vspace{.2in}
\noindent
{\bf Relation to block-sources.}
The generalized SV sources considered in this paper 
are also a generalization of ``block-sources" defined by Chor and Goldreich \cite{CG},
where the source is divided into several blocks such that each block has min-entropy at least $k$ conditioned 
on the value of the previous blocks.
Such a block-source can be thought 
as a generalized SV source where the adversary can generate each block (given previous blocks)
using any ``flat" distribution with support $2^k$.
Being a special case of generalized SV sources (defined here),
block-sources have another difference as well:
Since it is impossible to extract from a single block-source deterministically,
the common results regarding extraction from block-sources are about either seeded extractors (e.g.\ \cite{NZ}) or extraction from at least two independent block-sources (e.g.\ \cite{Rao}).

\vspace{.2in}
\noindent
{\bf Common randomness extractors.}
Common random bits, shared by distinct parties, constitute an important resource for distributed algorithms; common random bits can be used by the parties to synchronize the randomness of their local actions.
We may ask the question of randomness extraction in this setting too.
Assuming that the parties are provided with an imperfect source of common randomness, the question is whether perfect common randomness can be extracted from this source or not.

G\'acs and K\"orner~\cite{GacsKorner} and Witsenhausen~\cite{Witsenhausen} have looked at the problem of extraction of common random bits from a very special class of imperfect sources, namely i.i.d.\ sources. In this case, the \emph{bipartite} source available to the parties is generated as follows: In each time step, a pair $(A,B)$ with some predetermined distribution (known by the two parties) and independent of the past is generated; $A$ is revealed to the first party and $B$ is revealed to the second party. After receiving arbitrarily many repetitions of random variables $A$ and $B$, the two parties aim to extract a common random bit.  It is known that in this case, the two parties (who are not allowed to communicate) can generate a common random bit if and only if $A$ and $B$ have a common data~\cite{Witsenhausen}. This means that common randomness generation is possible if $A$ and $B$ can be expressed as $A = (A', C)$ and $B = (B', C)$ for a nonconstant common part $C$, i.e., there are nonconstant functions $f, g$ such that $C=f(A)=g(B)$.
 Observe that when a common part exists, common randomness can be extracted by the parties by applying the same extractor on the sequence of $C$'s. That is, the problem of common randomness extraction in the i.i.d.\ case is reduced to the problem of ordinary randomness extraction.
These results are obtained using a measure of correlation called~\emph{maximal correlation}.
The key feature of this measure of correlation that helps proving the above result is the \emph{tensorization property}, i.e., the maximal correlation between random variables $A$ and $B$ is equal to that of $A^n$ and $B^n$ for any $n$, where $A^n$ and $B^n$ denote $n$ i.i.d.\ repetitions of $A$ and $B$.

In this paper we  consider the problem of common randomness extraction from \emph{distributed SV sources} defined as follows.
In a distributed SV source, the adversary again has some multi-faceted dice,
but here, instead of a single number, a pair of numbers $(A, B)$ is written on each face.
As before, the set of values written on the faces of the dice is the same,
but the probabilities of face values may differ in different dice.
In each time step, the adversary depending on the results of the previous throws, picks a die and throws it.
If $(A, B)$ is the result of the throw, $A$ is given to the first party and $B$ to second party.
Thus, the two parties will observe random variables $A$ and $B$ whose joint distribution depends on the choice of die by the adversary.

Again consider the non-degenerate case
where all faces on all the dice of the adversary have positive probability.
We show that in this case, we can extract a common random bit from the distributed SV source
if and only if
it is possible to extract randomness from the common part of $A$ and $B$. That is, similar to the i.i.d.\ case, the problem of common randomness extraction from distributed SV sources is reduced to the problem of randomness extraction from non-binary generalized SV sources. Since by our results, we know when randomness extraction from generalized SV sources is possible, we obtain a  complete answer to the problem in the distributed case too.

This relation between the problem of common randomness extraction and the problem of randomness extraction from the common part holds in more general settings.
For example, it resolves the problem of common randomness extraction
from the following interesting distributed SV source.

\vspace{.2in}
\noindent
{\bf Example.} A concrete example of a distributed SV source is as follows. Let us start with the original source considered by Santha and Vazirani with two coins. Assume that the adversary chooses coin $S\in\{1,2\}$ (where coin 1 is biased towards heads and coin 2 is biased towards tails) and let the outcome of the throw of the coin be denoted by random variable $C$. The first party, Alice, is assumed to observe both the identity of the coin chosen by the adversary, i.e., $S$, and the outcome of the coin, which is $C$. The second party, Bob, observes the outcome of the coin $C$, but only gets to see the choice of the adversary with probability 0.99. That is, Bob gets $B=(C, \tilde S)$ where $\tilde S$ is the result of passing $S$ through a binary erasure channel with erasure probability 0.01. Here the common part of $A=(C,S)$ and $B=(C,\tilde S)$ is just $C$. Our result (Theorem~\ref{thm:main-distributed}) then implies that Alice and Bob cannot benefit from their knowledge of the actions of adversary, and should only consider the $C$ sequence. But then from the result of \cite{SanthaVazirani}, we can conclude that common random bit extraction is impossible in this example.

\vspace{.2in}
\noindent
{\bf Proof techniques.}
We briefly explain the techniques used in the proof of the above results.

Consider a deterministic randomness extractor that extracts one bit from a generalized SV source.
We can view this extractor as labeling the leaves of a rooted tree with zeros and ones.
Each sequence of dice throws corresponds to a path from the root to one of the leaves,
and at each node, the adversary has some limited control of which branch to take while moving from the root towards the leaves.
To prove the impossibility of randomness extraction, we need to show that either the minimum or the maximum of the probability of the output bit being zero, over all adversary's strategies, is far from $1/2$.
Our idea is to track these maximum and minimum probabilities in a recursive way,
i.e., to find these probabilities for any node of the tree in terms of these values for its children.  We then by induction show that for each node of the tree either the minimum probability or the maximum probability is far from $1/2$.

To be more precise, given a deterministic extractor, let $\alpha$ be the minimum probability of output bit being zero (over all strategies of the adversary). Similarly, let  $\beta$ be the maximum probability of output bit being zero (over all strategies of the adversary).
Then we show that under certain conditions, there exists a \emph{continuous} function $g(\cdot)$ on the interval $[0,1]$, such that $\beta\geq g(\alpha)$ and furthermore $g(1/2)> 1/2$. We prove $\beta\geq g(\alpha)$ inductively using the tree structure discussed above.
 This implies the desired impossibility result, as by the continuity of $g(\cdot)$, both $\alpha$ and $\beta$ cannot be close to $1/2$. For instance, for the binary SV source with two coins having  probability of heads respectively equal to $1/3$ and $2/3$, Figure~\ref{Fig1} shows a curve where $(\alpha, \beta)$ always lies above it.
 This curve is clearly isolated from $(1/2, 1/2)$.

\begin{figure*}[t]
\centering
\includegraphics[scale=0.4,angle=0]{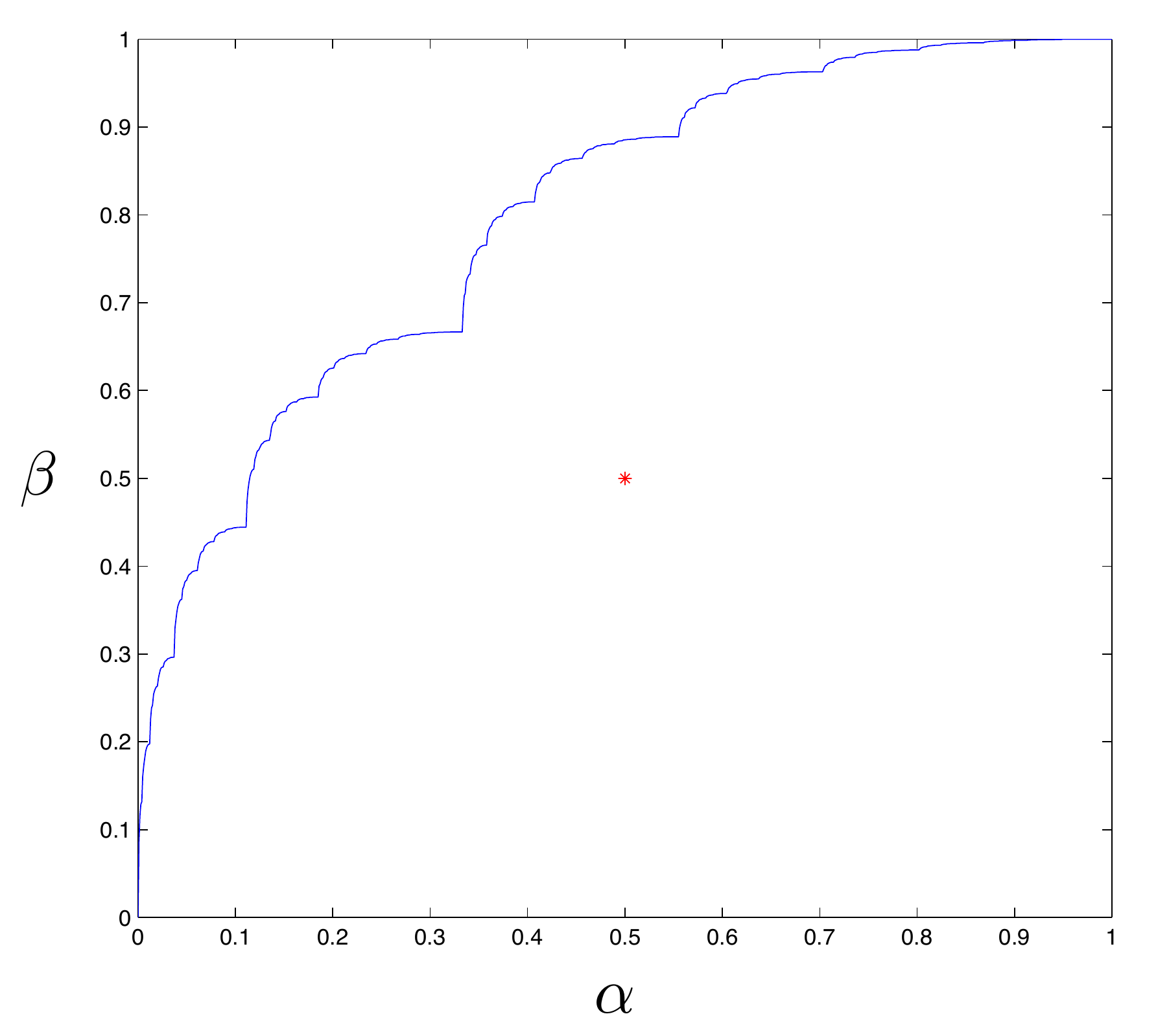}
\caption{Given any deterministic extractor, the pair $(\alpha, \beta)$ is above the curve specified in this figure, where
$\alpha$ and $\beta$ are the minimum and maximum value of probability of the output being zero that the adversary can achieve by choosing its strategy. The plot is for the binary SV source with two coins with probability of heads respectively equal to 1/3 and 2/3. The point $(1/2, 1/2)$ is specified by a red star in the figure.
To see how the curve is obtained, see Appendix~\ref{appendixC}, Corollary~\ref{thm:16}.}
\label{Fig1}
\end{figure*}

We follow similar ideas for proving our impossibility result for common randomness extraction from a distributed SV source; again we construct a continuous function, which somehow captures not only the minimum and maximum of the probability of the extracted common bit being zero, but also
the probability that the two parties agree on their extracted bits. The construction of this function is more involved in the distributed case; it has two terms one of which
is similar to the function in the non-distributed case, and the other is inspired by the definition of maximal correlation mentioned above.

To show the possibility of deterministic randomness extraction under certain conditions,
we try to use linear relations among the probability vectors associated with different dice
in order to define a martingale with anti-concentration properties,
and use the theory of stopping times for martingales and submartingales.

\vspace{.2in}
\noindent
{\bf Contributions to information theory.}
As mentioned above, the problem of common randomness extraction from i.i.d.\ sources has been studied in the information theory community. Then our work provides a generalization and an alternative proof of known results in the i.i.d.\ case.  In particular, we give a new proof of  Witsenhausen's result~\cite{Witsenhausen} on the impossibility of common randomness extraction from certain i.i.d.\ sources.

We also would like to point out that a generalized SV source as we define, is indeed an arbitrarily varying source (AVS)~\cite{Dobrusin1, Dobrusin2} with a causal adversary.  These sources are studied in the information theory literature from the point of view of source coding~\cite{Berger}.

%\vspace{.2in}\noindent
%{\bf Structure of the paper.} This paper is structured as follows: in Section \ref{sec:extraction-sv}, we define generalized SVi sources and provide sufficient and necessary conditions for randomness extraction from these sources. In Section \ref{sec:dist-sv}, we consider distributed general SV sources and study the possibility of common randomness extraction from these sources. Finally, in Section \ref{sec:future-work} we comment on some possible extensions of our results as future work. Appendices contain some of the technical details of the proofs.

\vspace{.2in}
\noindent
{\bf Notations.} In the rest of this section we fix some notations that will be used. The main results of the paper are discussed in the following two sections.

Probability spaces in this paper are all over finite sets which are denoted by calligraphic letters such as $\mathcal C$.  So a probability distribution over $\mathcal C$ is determined by numbers $p(c)$ for $c\in \mathcal C$. The random choice of $c\in \mathcal C$ with this distribution is denoted by $C$, i.e., $C=c$ with probability $p(c)$.

In this paper we also consider functions $X:\mathcal C\rightarrow \mathbb R$. Such a function can be thought of as a random variable $X=X(C)$.  We sometimes for simplicity use the notation $X(c)=x_c$.
The constant function $X:\mathcal C\rightarrow \mathbb R$ with $X(c)=1$ is denoted by $\mathbf 1_C$.
The expected value and variance of $X$ are denoted by $\E[X]$ and $\Var[X]$ respectively.
Given two such functions $X, Y:\mathcal C \rightarrow \mathbb R$ we define their inner product by
$$\langle X, Y\rangle:=\E[XY].$$
This inner product induces the norm
$\|X\|:=\langle X, X\rangle^{1/2}=\E[X^2]^{1/2}.$

We sometimes have several distributions over the same set $\mathcal C$ which are indexed by elements $s\in \mathcal S$. In this case to avoid confusions, the expectation value, variance, inner product, and norm are specified by a subscript, i.e., respectively by $\E_{(s)}, \Var_{(s)},$ $\langle\cdot , \cdot \rangle_{(s)}$, and $\|\cdot\|_{(s)}$. The \emph{uniform distribution} is specified by a \emph{star} subscript, i.e,
$$\E_*[X]=\frac{1}{|\mathcal C|} \sum_c x_c,$$
and
$\langle X, Y\rangle_*:= \E_*[XY]$, and $\|X\|_*=\E_*[X^2]^{1/2}$.

For simplicity of notation a sequence $C_1, \dots, C_n$ of (not necessarily i.i.d.)\ random variables is denoted by $C^n$. Similarly for $c_1,\dots, c_n\in \mathcal C$ we use $c^n=(c_1, \dots, c_n)$. We also use the notation
$c_{[k: k+\ell]} = (c_k, c_{k+1},\dots, c_{k+\ell})$.

\section{Randomness extraction from generalized SV sources}\label{sec:extraction-sv}
\begin{definition}[Generalized SV source]
Let $\mathcal C$ be a finite alphabet set.
Consider a \emph{finite} set of distributions over $\mathcal C$ indexed by a set $\mathcal S$. That is, assume that for any $s \in \mathcal S$ we have a distribution over $\mathcal C$ determined by numbers $p_s(c)$ for  all $c\in \mathcal C$.
A sequence $C_1, C_2, \cdots$ of random variables, each over alphabet set $\mathcal C$,
is said to be a \emph{generalized SV source} with respect to distributions $p_s(c)$, if the sequence is generated as follows:
Assume that $C_1, \ldots, C_{i-1}$ are already generated.
In order to determine $C_i$, an adversary chooses $S_i = s_i \in \mathcal S$, depending only\footnote{We can allow for the adversary to choose $s_i$ depending both on $C_1, \ldots, C_{i-1}$ and on $S_1, \ldots, S_{i-1}$, but this relaxation is not important, since it is only the marginal distribution of $p(c_1, c_2, \cdots, c_n)$ that matters to us.} on $C_1, \ldots, C_{i-1}$.  Then $C_i$ is sampled from the distribution $p_{s_i}(c).$
\end{definition}

We can think of specifying $s$ as choosing a particular multi-faceted die, and $c$ as the facet that results from throwing the die.
The joint probability distribution of random variables $C_1, \dots, C_n$ and $S_1, \dots, S_n$ in a generalized SV source factorizes as follows:
$$p(c_1, c_2, \cdots, c_n, s_1, s_2, \cdots, s_n)=q(s_1)p_{s_1}(c_1)q(s_2|c_1)p_{s_2}(c_2)\cdots q(s_n|c_1\cdots c_{n-1})p_{s_n}(c_n),$$
where $q(s_i|c_1\cdots c_{i-1})$ describes the action of the adversary at time $i$. Here, first the adversary chooses $S_1=s_1$ with probability $q(s_1)$, and then $C_1=c_1$ is generated with probability $p_{s_1}(c_1)$. Then the adversary chooses $S_2=s_2$ with probability $q(s_2|c_1)$ and then $C_2=c_2$ is generated with probability $p_{s_2}(c_2)$, and so on.

Generalized SV sources can be alternatively characterized as follows:
Given $i$ and $C_1=c_1, \ldots, C_{i-1}=c_{i-1}$, the distribution of $C_i$ should be a convex combination of the set of $|\mathcal S|$ distributions $\{p_s(\cdot): s \in \mathcal S\}$.

We emphasize that even after fixing distributions $p_s(c)$, the generalized SV source (similar to ordinary SV sources) is not
a fixed source, but rather a class of sources. This is because in each step $s_i$ is chosen arbitrarily by the adversary as a (probabilistic) function of $C_1, \ldots, C_{i-1}$. Nevertheless, once we fix adversary's strategy, the generalized SV source is fixed in that class of sources.

\begin{definition}[Deterministic extraction]
We say that deterministic randomness extraction from the generalized SV source determined by distributions
$p_s(c)$ is possible if for every $\epsilon > 0$ there exist $n$ and $\Gamma_n: \mathcal C^n \rightarrow \{0,1\}$
such that for every strategy of the adversary,
the distribution of $\Gamma_n(C^{n})$ is $\epsilon$-close, in total variation distance, to the uniform distribution.
That is, independent of adversary's strategy, $\Gamma_n(C^n)$ is an almost uniform bit.
\end{definition}

In the following we present a necessary condition and separately a sufficient condition for the existence of deterministic extractors for generalized SV sources. In the non-degenerate case, i.e., when $p_s(c) > 0$ for all $s, c$, these two conditions coincide. Thus we fully characterize the possibility of deterministic randomness extraction from generalized SV sources in the non-degenerate case.

\subsection{A sufficient condition for the existence of randomness extractors}

In this subsection we prove the following theorem.

\begin{theorem}\label{thm:martingale}
Consider a generalized SV source with alphabet $\mathcal C$, set of dice $\mathcal S$, and probability distributions $p_s(c)$.
Suppose that there exists $\psi:\mathcal C\rightarrow \mathbb R$ such that for every $s \in \mathcal S$ we have
$\E_{(s)}[\psi(C)]=0$ and $\Var_{(s)}[\psi(C)]>0$,
where $\E_{(s)}$ and $\Var_{(s)}$ are expectation and variance with respect to the distribution $p_s(\cdot)$.
Then randomness can be extracted from this SV source.
\end{theorem}

Observe that if $p_s(c)>0$ for all $s, c$, then this theorem can equivalently be stated as follows:
Thinking of each distribution $p_s(\cdot)$ as a point in the probability simplex, if the
 convex hull of the set of points $\{p_s(\cdot): s \in \mathcal S \}$ in the probability simplex does not have full dimension, then deterministic randomness extraction is possible. For instance if $|\mathcal S|<|\mathcal C|$ this condition is always satisfied and then we can deterministically extract randomness.

\begin{proof}[Proof of Theorem~\ref{thm:martingale}]
Pick a sufficiently large (but constant) number $M$.
Define random variables $X_1, \ldots, X_n$ and $Y_0, \ldots, Y_n$ inductively as follows:
Let $Y_0 = 0$, and for $i = 1, \ldots, n$,  define $Y_i = Y_{i-1} + X_i$ where $X_i = \psi({C_i})$.
Observe that by our assumption we have $\E[X_i| X_1, \dots, X_{i-1}]=0$,  so $Y_0, \dots, Y_n$ forms a martingale.

Let $\tau$ be the first time $t$ such that $|Y_t| \ge M$;
if no such $t$ exists, define $\tau = n$.
Clearly, $\tau$ is a stopping time for the martingale.
Now define the extracted bit to be $1$ if $Y_\tau \ge M$; otherwise define it to be 0. We show that this is a true random bit extractor.

Let $v = \min_s \Var_{(s)}[\psi] > 0$.
Define $Z_i = Y_i^2 - i v$.
We claim that $Z_i$ is a submartingale with respect to $X_1, \ldots, X_n$. To show this we compute
\begin{align*}
\E [Z_i | X_1, \ldots, X_{i-1}] & = \E\big[ (X_i + Y_{i-1})^2 - i v \big| X_1, \ldots, X_{i-1}\big] \\
& = \E\big[ (Y_{i-1}^2 - (i - 1) v) + (X_i^2 - v) + 2 X_i Y_{i-1} \big| X_1, \ldots, X_{i-1}\big] \\
& \ge Z_{i-1}.
\end{align*}
Here we used $Z_{i-1} = Y_{i-1}^2 - (i - 1) v,$ and
$$\E [X_i Y_{i-1} | X_1, \ldots, X_{i-1}] = Y_{i-1} \E[X_i | X_1, \ldots, X_{i-1}] = 0,$$
and that by the law of total variance
$$\E[X_i^2 | X_1, \ldots, X_{i-1}] = \Var[\psi(C_i)| X_1, \ldots, X_{i-1}] \ge \Var[\psi(C_i) | X_1, \ldots, X_{i-1}, S_i] \ge v.$$
Therefore by optional stopping theorem for submartingales,
we have
$$\E [Z_\tau] \ge \E[Z_0] = 0,$$
or equivalently
$$\E [Y_\tau^2] \ge  v \E[\tau].$$

Let $m = \max_c |\psi(c)|$. Then, by the definition of $\tau$ we have $|Y_\tau| \le M + m$.
Therefore, $$\E[\tau] \le \frac{\E[Y_\tau^2]}{v} \le \frac{(M+m)^2}{v}.$$
Hence by the Markov inequality we have
$$\Pr[\tau = n] \le \frac{(M+m)^2}{vn} = O\big(\frac 1 n\big).$$
This means that
$$\Pr\big[Y_\tau \in [M, M + m) \cup (-M - m, -M]\big] = 1 - O\big(\frac 1 n\big).$$
On the other hand, for the martingale $Y_0, Y_1, \ldots$, we have $\E [Y_\tau] =\E[Y_0]= 0$.
Together with $|Y_\tau| \le M + m$, this implies
$$\frac{M}{2M+m} + O\big(\frac 1 n\big) \le \Pr[Y_\tau \in [M, M + m)] \le \frac{M+m}{2M+m} + O\big(\frac 1 n\big).$$
Therefore, the extracted bit has sufficiently small bias as $M, n$ are chosen sufficiently large.
This is because $m=\max_c |\psi(c)|$ is a constant, independent of $M$ and $n$.

\end{proof}

\begin{remark}
Note that the extractor constructed in the above proof is \emph{explicit}.
Moreover, although we have only mentioned how to extract a single bit,
the analysis shows that for arbitrarily small (but \emph{constant}) bias, one can extract \emph{linearly} many bits each having at most that bias given the previous bits. This can be done by partitioning the sequence into a linear number of blocks. One bit is extracted from each block. Each produced bit is almost uniform, given the past blocks and hence given the past produced bits. Thus, the bits are almost uniform and almost mutually independent.
\end{remark}

\begin{remark}
Note that we could have chosen $M = \Theta(n^{1/3})$ in the above proof.
Then the analysis would have shown that the bias is polynomially small, namely a bias of $\Theta(n^{-1/3})$.
\end{remark}

\subsection{A necessary condition for the existence of randomness extractors}

The main result of this subsection is the following theorem.

\begin{theorem}\label{thm:necessary-condition}
Consider a generalized SV source with alphabet $\mathcal C$, set of dice $\mathcal S$, and probabilities $p_s(c)$.
Suppose that there is no non-zero function $\psi:\mathcal C\rightarrow \mathbb R$ such that for all $s\in \mathcal S$ we have $\E_{(s)}[\psi(C)]=0$. Then deterministic randomness extraction from this generalized SV source is impossible.
\end{theorem}

Again, let us consider the case where $p_s(c)>0$ for all $s, c$. In this case $\psi$ being non-zero is equivalent to $\Var_{(s)}[\psi]>0$ for all $s$. Then comparing to Theorem~\ref{thm:martingale} we find that the necessary and sufficient condition for the possibility of deterministic extraction is the existence of a non-zero $\psi$ with $\E_{(s)}[\psi]=0$.

In Appendix~\ref{app:proof1-necessary} we give a proof of this theorem based on ideas in~\cite{RVW}. Here we present another proof whose ideas will be used in the distributed case too.

\begin{proof}[Proof of Theorem \ref{thm:necessary-condition}]
A deterministic randomness extraction algorithm corresponds to a subset $\mathcal I\subseteq \mathcal C^n$ such that the extracted bit is $0$ if the observed $c^{n}$ is in $\mathcal I$, and is $1$ otherwise. For any $n$, and any such $\mathcal I\subseteq \mathcal C^n$, let $\alpha(\mathcal I)$ and $\beta(\mathcal I)$  respectively be the minimum and maximum of the probability of output $0$ over all strategies of the adversary, i.e.,
$$\alpha(\mathcal I):=\min \Pr[C^{n}\in \mathcal I], \qquad \beta(\mathcal I):=\max \Pr[C^{n}\in \mathcal I],$$
where minimum and maximum are taken over adversary's strategies.

Fix a deterministic algorithm for randomness extraction. To prove the theorem we need to show that for every such $\mathcal I$, either $\alpha(\mathcal I)$ or $\beta(\mathcal I)$ is far from $1/2$. The numbers $\alpha(\mathcal I), \beta(\mathcal I)$ can be computed recursively as follows. For every $c\in \mathcal C$, let $\mathcal I_c:=\{c_{[2: n]}:\, (c, c_{[2:n]})\in \mathcal I\}$. Note that $\mathcal I_{c}$ is a subset of $\mathcal C^{n-1}$ for which $\alpha(\mathcal I_c)$ is defined. We claim that
$$\alpha(\mathcal I)=\min_s \sum_c p_s(c)\alpha(\mathcal I_c)= \min_s \E_{(s)}[\alpha(\mathcal I_C)].$$
To verify this, suppose that the adversary in the first step chooses $s_1=s$. Then $C_1=c$ occurs with probability $p_{s}(c)$. Assuming $C_1=c$, the final extracted bit is equal to $0$ if $(C_2, \dots, C_n)\in \mathcal I_c$. Since, by definition,  the minimum of the probability of this latter event is $\alpha(\mathcal I_c)$, the (unconditional) probability of the extracted bit being 0 is equal to $\sum_c p_s(c)\alpha(\mathcal I_c)$. Taking the minimum of this expression over all $s_1=s$ gives $\alpha(\mathcal I)$.
We similarly have
$$\beta(\mathcal I)= \max_s \E_{(s)} [\beta(\mathcal I_C)].$$

By the above discussion to compute $\alpha(\mathcal I)$ and $\beta(\mathcal I)$ for $\mathcal I\subseteq \mathcal C^{n}$ it suffices to compute these numbers for subsets of $\mathcal C^{n-1}$. Thus the functions $\alpha(\cdot)$ and $\beta(\cdot)$ can be computed recursively. The above recursive procedure can be understood as assigning two values to each node of the tree associated with the extractor, as described in the ``proof techniques" subsection of the introduction.

Let $\Phi_n$ be the set of pairs $(\alpha(\mathcal I), \beta(\mathcal I))$ for all subsets $\mathcal I\subseteq \mathcal C^n$. In other words, for $n\geq 1$ define
$$\Phi_n:= \big\{(\alpha(\mathcal I), \beta(\mathcal I)):\,   \mathcal I\subseteq \mathcal C^n\big\}.$$
Also let
$$\Phi_0=\{(0,0), (1,1)\}.$$
Observe that $\Phi_0$  corresponds to the case when there is no SV source to look at, and the deterministic extractor outputs a constant bit.
Now by the above discussion, $\Phi_n$ is indeed the set of pairs $(x, y)$ for which there exist $X, Y:\mathcal C\rightarrow \mathbb R$ such that  $(X(c), Y(c))=(x_c, y_c)\in \Phi_{n-1}$ for every $c\in \mathcal C$, and that
\begin{align}
x&=\min_s \E_{(s)}[X]=\min_s \sum_{c}p(c|s)X(c), \nonumber
\\y&=\max_s \E_{(s)}[Y]=\min_s \sum_{c}p(c|s)Y(c).\label{eq:xy-min-max}
\end{align}
 A full characterization of the set $\Phi_n$ for the original binary SV source is given in Appendix~\ref{appendixC}.

Suppose that $g:[0,1]\rightarrow \mathbb R$ is a function that satisfies the followings:
\begin{itemize}
\item $g$ is continuous and \emph{monotone},
\item we have
\begin{align}\label{eq:g-0-1-12}
g(0)=0,\qquad g(1)=1, \qquad g(1/2)>1/2,
\end{align}
\item and for all $X:\mathcal C\rightarrow [0,1]$ we have
$$\max_s \E_{(s)}[g(X)] \geq \min_{s'} g\big(  \E_{(s')} [X] \big),$$
or equivalently
\begin{align}\label{eq:inequality-cs-2}
\max_{s,s'} \E_{(s)} [g(X)] - g\big(  \E_{(s')}[X] \big)\geq 0.
\end{align}
\end{itemize}
Then we claim that $\beta(\mathcal I)\geq g(\alpha(\mathcal I))$. To prove this, it suffices to show that for all $(x, y)\in \Phi_n$ we have $y\geq g(x)$. The latter statement can be proved by induction on $n$. The base of induction, $n=0$, follows from $g(0)=0$ and $g(1)=1$. Assuming that $(x, y)\in \Phi_n$ is obtained from~\eqref{eq:xy-min-max} for $(x_c, y_c)\in \Phi_{n-1}$, by the induction hypothesis we have $y_c\geq g(x_c)$, and then
\begin{align*}
g(x) & = g\big(\min_s \E_{(s)}[X]\big)\\
& = \min_s g\big(\E_{(s)}[X]\big)\\
& \leq \max_s \E_{(s)}[g(X)]\\
& \leq \max_s \E_{(s)}[Y]\\
& = y.
\end{align*}
Here in the second line we use the monotonicity of $g$, and in the fourth line we use the induction hypothesis.

If such a function $g$ with the above properties exists, then $\alpha(\mathcal I)$ and $\beta(\mathcal I)$ cannot both be arbitrary close to $1/2$. To verify this, note that  $\beta(\mathcal I)\geq g(\alpha(\mathcal I))$, so if $(\alpha(\mathcal I), \beta(\mathcal I))\simeq (1/2, 1/2)$, by the continuity of $g$ we have $1/2\gtrsim g(1/2)$. This is a contradiction since $g(1/2)>1/2$.
As a result, we only need to prove the existence of the function $g$.

Let $f:[0,1]\rightarrow \mathbb R$ be a smooth function such that $f(1/2)>0$ and $f(0)=f(1)=0$.
We show that the function $g_\epsilon$ defined by
\begin{align}g_\epsilon(x):=x+ \epsilon f(x),\label{eqn:defi:g_ep}\end{align}
for sufficiently small $\epsilon>0$, satisfies the desired properties. Verification of~\eqref{eq:g-0-1-12} is easy. For the monotonicity of $g_\epsilon$, note that since $f$ is smooth, there is a uniform upper bound $|f'(x)|\leq M$ on the derivative of $f$. Then for $\epsilon<1/M$, the function $g_\epsilon$ is monotone. It remains to show~\eqref{eq:inequality-cs-2}.

Define
$$\mathcal T:= \big\{T:\mathcal C\rightarrow [0,1]:\, \|T\|_*=1,\, \, \E_*[T]=0\big\}.$$
For every $T\in \mathcal T$ we have
$$\max_{s,s'} \E_{(s)}[T] -\E_{(s')}[T]>0,$$
because otherwise we would obtain a non-constant function whose expectation is independent of $s$, which is in contradiction with our assumption in the statement of the theorem.
Therefore, using the compactness of $\mathcal T$, there is $\Delta>0$ such that
$$\max_{s,s'} \E_{(s)}[T] -\E_{(s')}[T]>\Delta,\qquad\qquad \forall T\in \mathcal T.$$

Let $X:\mathcal C\rightarrow  [0,1]$ be an arbitrary function. Then, letting $x=\mathbb{E}_*[X]$ and $r=\sqrt{\Var_*[X]}\geq 0$ we get that
$$X= x \mathbf 1_C + rT=x+rT,$$
for some $T\in \mathcal T$, i.e., $x_c=x+ rt_c$ for all $c\in \mathcal C$.
From equation \eqref{eqn:defi:g_ep} we have
\begin{align*}
\max_{s,s'} &\, \E_{(s)}[g_\epsilon(X)] - g_\epsilon\big(  \E_{(s')}[X] \big) \\
&= \max_{s, s'}\, r \big(\E_{(s)}[T]- \E_{(s')}[T]\big) + \epsilon\Big(\E_{(s)}[f(x+rT)] - f\big(x+r\E_{(s')}[T] \big)   \Big).
\end{align*}

For every $0\leq x,  y \leq 1$ there is some $z$ (between $x$ and $y$) such that $f(y)=f(x)+ (y-x)f'(z)$.
Using the upper bound $M$ on the derivative of $f$ we obtain
$$f(x)-M|y-x|\leq f(y)\leq f(x)+M|y-x|.$$
Therefore, using the fact that $|t_c|\leq \sqrt{|\mathcal C|} \le |\mathcal C|$ (implied by $\|T\|_*=1$), we have
\begin{align*}
\max_{s,s'} &\, \E_{(s)}[g_\epsilon(X)] - g_\epsilon\big(  \E_{(s')}[X] \big) \\
&\geq \max_{s, s'}\, r \big(\E_{(s)}[T]- \E_{(s')}[T]\big) + \epsilon\Big(\E_{(s)}[f(x)-rM|T|] - \big(f(x) + r M \E_{(s')}[|T|]\big) \Big)\\
&\geq \max_{s, s'}\, r \big(\E_{(s)}[T]- \E_{(s')}[T]\big) - 2\epsilon r M |\mathcal C|\\
&\geq r(\Delta - 2\epsilon M |\mathcal C|).
\end{align*}
which is strictly positive if $\epsilon<\Delta/(2M|\mathcal C|)$. Then the function $g_\epsilon$ for
$$\epsilon<\min\{  1/M,  \Delta/(2M|\mathcal C|)\},$$
has all the desired properties.
\end{proof}

\begin{corollary}\label{cor:necessary}
Consider a generalized SV source with alphabet $\mathcal C$, set of dice $\mathcal S$, and probabilities $p_s(c)$. Let $\mathcal{S}'$ be a subset of $\mathcal{S}$ and let $\mathcal{C}'$ be the set of all $c$ for which there exists some $s\in \mathcal{S}'$ such that $p_{s'}(c)>0$. Suppose that there is no non-zero function $\psi:\mathcal C\rightarrow \mathbb R$ such that (i) $\psi$ is zero on $\mathcal C - \mathcal{C}'$, and (ii) for all $s\in \mathcal S'$ we have $\E_{(s)}[\psi(C)]=0$. Then deterministic randomness extraction from this generalized SV source is impossible.
\end{corollary}

%********************************************
\section{Distributed SV sources}\label{sec:dist-sv}

Distributed SV sources can be defined similarly to generalized SV sources except that in this case, the outcome in each time step is a pair that is distributed between two parties. 
\begin{definition}\label{def:dist-SV}
Fix finite sets $\mathcal A, \mathcal B$, $\mathcal S$. Let $p_s(ab)$ define a probability distribution over $\mathcal A\times \mathcal B$ for any $s\in \mathcal S$. The distributed SV source with respect to distributions $p_{s}(ab)$ is defined as follows.
The adversary in each time step $i$, depending on the previous outcomes $(A_1, B_1)=(a_1, b_1), \dots, (A_{i-1}, B_{i-1})=(a_{i-1}, b_{i-1})$ chooses some $S_i = s_i$. Then $(A_i, B_i)=(a_i, b_i)$ is sampled from the distribution $p_{s_i}(a_ib_i)$. The sequence of random variables $(A_1, B_1), (A_2, B_2), \dots$, is called a distributed SV source.
\end{definition}
Here
we assume that the outcomes of this SV source are distributed between two parties, say Alice and Bob. That is, in each time step $i$, $A_i$ is revealed to Alice and $B_i$ is revealed to Bob. So Alice receives the sequence $A_1, A_2, \dots$, and Bob receive the sequence $B_1, B_2, \dots$.

In this section we are interested in whether two parties can generate a common random bit from distributed SV sources.
To be more precise, let us first define the problem more formally.

\begin{definition}
We say that common randomness can be extracted from the distributed SV source $(A_1, B_1), (A_2, B_2), \ldots$ if for every $\epsilon >0$ there is $n$ and functions $\Gamma_n:\mathcal A^n\rightarrow \B$ and $\Lambda_n:\mathcal B^n\rightarrow \B$ such that for every strategy of adversary, the distributions of $K_1=\Gamma_n(A^{n})$ and $K_2=\Lambda_n(B^n)$ are $\epsilon$-close (in total variation distance) to uniform distribution, and that $\Pr[K_1\neq K_2]<\epsilon$.
\end{definition}

In the above definition we considered only deterministic protocols for extracting a common random bit. We could also consider probabilistic protocols where $\Gamma_n$ and $\Lambda_n$ are random functions depending on \emph{private} randomnesses of Alice and Bob respectively.
More precisely, we could take $K_1=\Gamma_n(A^n, R_1)$ and $K_2=\Lambda_n(B^n, R_2)$ with the above conditions on $K_1, K_2$, where $R_1$ and  $R_2$ are private randomnesses of Alice and Bob respectively, which are independent of the SV source and of each other.
Nevertheless, if a common random bit can be extracted with probabilistic protocols, then common randomness extraction with deterministic protocols is also possible.

\begin{lemma}\label{lemma1}
In the problem of common random bit extraction, with no loss of generality we may assume that the parties do not have private randomness.
\end{lemma}
The proof of this lemma is given in Appendix \ref{app:lemma1}.

\subsection{Maximal correlation}

Let us first consider the problem of common randomness extraction in a simpler case where there is no adversary (in the i.i.d.\ case). That is, let us assume that we have only one distribution $p(ab)$, and Alice and Bob in each time $i$ receive samples $A_i$ and $B_i$ from this distribution. The question of the possibility of common randomness extraction can be raised in this case too.

Witsenhausen~\cite{Witsenhausen} used a measure of correlation called maximal correlation to prove a necessary and sufficient condition for the possibility of common randomness extraction from i.i.d.\ sources.

\begin{definition}[Maximal correlation]\label{def3}
The maximal correlation of random variables $A$ and $B$ with joint distribution $p(ab)$ denoted by $\rho(A, B)$ is defined by
\begin{align}\label{eq:max-correlatoin}
\rho(A, B):=& ~~\max \quad ~~~\E[XY],
\\&\text{\rm{subject to:} } \E[X]=\E[Y]=0, \nonumber\\
& \qquad \qquad \quad \E[X^2]=\E[Y^2]=1,\nonumber
\end{align}
where the maximum is taken over all functions $X:\mathcal{A}\rightarrow \mathbb{R}$,  $Y:\mathcal{B}\rightarrow \mathbb{R}$.
\end{definition}

Maximal correlation has the intriguing property that if $(A^n, B^n)$ is $n$ i.i.d.\ copies of $(A, B)$, then
$\rho(A^n, B^n)=\rho(A, B)$. Moreover, maximal correlation does not increase under local stochastic maps~\cite{Witsenhausen}.

From the definition and using Cauchy-Schwarz inequality it is not hard to verify that $0\leq \rho(A, B)\leq 1$. Further,  $\rho(A, B)=0$ if and only if $A, B$ are independent. To characterize the other extreme case $\rho(A, B)=1$ we need the notion of common data.

\begin{definition}
We say that $A, B$ have common data if there are non-constant functions $\Gamma(A)$ and $\Lambda(B)$, with arbitrary but the same images, such that $\Gamma(A)=\Lambda(B)$ with probability one. 
\end{definition}
Thus $A$ and $B$ have common data if Alice and Bob, having access to $A$ and $B$ respectively, can compute the same non-trivial data (i.e., $\Gamma(A)=\Lambda(B)$) without communication. We have $\rho(A, B)=1$ if and only if $A, B$ have common data.

\begin{theorem}\cite{Witsenhausen} \label{thm:Witsenhausen}
A common random bit can be extracted from i.i.d.\ copies of $A, B$ if and only if $\rho(A, B)=1$.
\end{theorem}

Here, we give an alternative proof of this theorem whose ideas will be used later. This proof of Witsenhausen's theorem can also be of independent interest.

\begin{proof}
If $\rho(A, B)=1$, then $A, B$ have common data as defined above, and a common random bit can be extracted from that common data
by standard randomness extractors for i.i.d.\ sources.

For the other direction, suppose that $\rho(A, B)=\rho<1$, and that we can extract one bit of common randomness from $A, B$. By Lemma~\ref{lemma1} we
may assume that Alice and Bob's strategies for extracting common randomness are deterministic. That is, we may assume that there are subsets $\mathcal I\subseteq \mathcal A^n$ and $\mathcal J\subseteq \mathcal B^n$ such that Alice's extracted bit is $K_1=0$ if $A^n\in \mathcal I$ and Bob's extracted bit $K_2=0$ if $B^n\in \mathcal J$, and that $K_1, K_2$ are equal with high probability, and their distributions are close to uniform distribution over $\B$.

Let us define
\begin{align*}
\alpha(\mathcal I)&:=\Pr[A^n\in \mathcal I],\\
\beta(\mathcal J)&:=\Pr[B^n\in \mathcal J],\\
\gamma(\mathcal I, \mathcal J)&:=\Pr[A^n\in \mathcal I~\&\  B^n\in \mathcal J].
\end{align*}
Then by assumption these three number are all close to $1/2$.

For every $a\in \mathcal A$ and $b\in \mathcal B$ define $\mathcal I_a=\{a_{[2:n]}:\, (a, a_{[2:n]})\in \mathcal I\}$ and $\mathcal J_b:=\{b_{[2:n]}:\, (b, b_{[2:n]})\in \mathcal J\}$. Then as in the proof of Theorem~\ref{thm:necessary-condition} the numbers $\alpha(\mathcal I),\beta(\mathcal J)$ and $\gamma(\mathcal I, \mathcal J)$ can be computed recursively:
\begin{align*}
\alpha(\mathcal I)&=\sum_a p(a)\alpha(\mathcal I_a)=\E[\alpha(\mathcal I_A)],
\\\beta(\mathcal J)&=\sum_b p(b)\beta(\mathcal J_b)=\E[\beta(\mathcal J_B)],
\\\gamma(\mathcal I, \mathcal J)&=\sum_{a,b} p(a,b)\gamma(\mathcal I_{a}, \mathcal{J}_{b})=\E[\gamma(\mathcal I_A, \mathcal J_B)].
\end{align*}

For $n\geq 1$, let $\Phi_n$ be the set of triples $(\alpha(\mathcal I), \beta(\mathcal J), \gamma(\mathcal I, \mathcal J))$ for all $\mathcal I\subseteq \mathcal A^n$ and $\mathcal J\subseteq \mathcal B^n$.
Also let
$$\Phi_0=\big\{\mathbf{e}_0=(1, 1, 1), \mathbf{e}_1=(1, 0, 0), \mathbf{e}_2=(0, 1, 0), \mathbf{e}_3=(0, 0, 0)\big\}.$$
Observe that $\Phi_0$ corresponds to deterministic strategies of Alice and Bob that determine $K_1, K_2$ without looking at any random source. If, for instance,  Alice always outputs $K_1=0$ and Bob always outputs $K_2=1$, then
$\Pr[K_1=0]=1$, $\Pr[K_2=0]=0$, and $\Pr[K_1=K_2=0]=0$. This gives the triple $\mathbf e_1=(1,0,0)$ in $\Phi_0$.

Now since by the above discussions the numbers $\alpha(\mathcal I), \beta(\mathcal J)$ and $\gamma(\mathcal I, \mathcal J)$ can be computed recursively, the sets $\Phi_n$ can be characterized recursively too. Indeed, $\Phi_n$ for $n\geq 1$ is the set of triples $(x, y, z)$ for which there exist functions $X(a)=x_a, Y(b)=y_b$ and $Z(ab)=z_{ab}$ such that for all $(a, b)$ we have $(x_a, y_b, z_{ab})\in \Phi_{n-1}$ and that
\begin{align}\label{eq:xXyYzZ}
x=\E[X],\qquad y=\E[Y], \qquad z=\E[Z].
\end{align}

Let us define the function $f:[0,1]^2\rightarrow \mathbb R$ by
\begin{align}
f(x,y,z):=(x+y)\rho-2z+2xy -(x^2+y^2)\rho.
\label{eqn:def:fxyz}
\end{align}
We claim that $f(\alpha(\mathcal I), \beta(\mathcal J), \gamma(\mathcal I, \mathcal J))\geq 0$. Assuming this, we conclude that $\alpha(\mathcal I), \beta(\mathcal J)$ and $\gamma(\mathcal I, \mathcal J)$ cannot all be close to $1/2$ because $f$ is continuous and
$$f(1/2, 1/2, 1/2)=-\frac{1-\rho}{2}<0.$$

To prove our claim it suffices to show that $f(x, y, z)\geq 0$ for all $(x, y, z)\in \Phi_n$, which itself can be proved by induction on $n$. The base of induction, $n=0$, follows from $f(\mathbf e_\ell)\geq 0$ for $0\leq \ell\leq 3$. Now suppose that $(x, y, z)\in \Phi_n$ is obtained from functions $X, Y, Z$ as above that satisfy~\eqref{eq:xXyYzZ}. By the induction hypothesis for every $(a, b)$ we have $f(X(a), Y(b), Z(a b))\geq 0$. Then
to prove $f(x, y, z)\geq 0$,  it suffices to show that
$$f(x, y, z)\geq \E[f(X, Y, Z)].$$
Using~\eqref{eq:xXyYzZ},   we need to show that
$$f(\E[X], \E[Y], \E[Z])\geq \E[f(X, Y, Z)].$$
Using the definition of the function $f(\cdot)$ in equation \eqref{eqn:def:fxyz}, and by expanding both sides and canceling the linear terms, we need to show that
$$2\E[X]\E[Y] -\rho(\E[X]^2 + \E[Y]^2)\geq 2\E[XY] - \rho(\E[X^2] + \E[Y^2]).$$
Let us define $X'=X-\E[X]$ and $Y'=Y-\E[Y]$. Then, expressing the above inequality in terms of $X', Y'$ we need to show that
$$2\E[X'Y']\leq \rho(\E[X'^2] + \E[Y'^2]).$$
This inequality is a consequence of the definition of $\rho=\rho(A, B)$ because $\E[X']=\E[Y']=0$ and then
\begin{align*}
\E[X'Y']& \leq  \rho \sqrt{\E[X'^2]\E[Y'^2]} \leq \frac{1}{2} \rho(\E[X'^2] + \E[Y'^2]).
\end{align*}
\end{proof}

\subsection{Common data}

In the previous subsection we briefly discussed the notion of common data and recalled that $\rho(A, B)=1$ if and only if common data exists. To state our result, however, we need a more precise characterization of common data.

Suppose that $A, B$ have a common data, meaning that there are non-trivial functions $\Gamma(A)$ and $\Lambda(B)$ such that $\Gamma(A)=\Lambda(B)$. Let $\mathcal C$ be the images of these functions. For any $c\in \mathcal C$ define $\mathcal A_c=\Gamma^{-1}(c)$ and $\mathcal B_c=\Lambda^{-1}(c)$. Given the fact that $\Gamma(A)=\Lambda(B)$ always holds, then for every $c\neq c'$ and $(a, b)\in \mathcal A_c\times \mathcal B_{c'}$ we must have $p(ab)=0$.

To understand this more precisely consider a bipartite graph $\mathcal G$ on the vertex set $\mathcal A\cup \mathcal B$ with an edge between $(a, b)$ if $p(ab)\neq 0$. Then by the above observation, the existence of common data implies that the graph $\mathcal G$ is disconnected (and also at least two of the connected compoenents are not singletons); if $c\neq c'$ then there is no edge between vertices in $\mathcal A_c\cup \mathcal B_c$ and $\mathcal A_{c'}\cup \mathcal B_{c'}$.

Conversely, if $\mathcal G$ is disconnected (and also at least two of the connected components are not singletons) then common data exists; letting $\mathcal C$ be the sets of connected components, and defining $\Gamma(a), \Lambda(b)$ be the index of the connected component to which $a, b$ belong, we have $\Gamma(A)=\Lambda(B)$. As a result, $\rho(A, B)=1$ if and only if $\mathcal G$ is disconnected (and at least two of the connected components are not singletons).

We summarize the above discussion in the following lemma.

\begin{lemma}\label{lem:bipartite-graph}
Let $C$ be the random variable associated to the index of the connected component of $\mathcal G$ to which $(A, B)$ belong. Then $C$ can be computed as a function of $A$ or $B$ individually. Moreover, any common data of $A, B$ is a function of $C$, and $\rho(A, B)=1$ if and only if $C$ is non-trivial (i.e., $\mathcal G$ has at least two non-singleton connected components).
\end{lemma}

\begin{example}\label{ex9}
Consider the following joint distribution on $\mathcal A\times \mathcal B$ where $\mathcal A=\mathcal B=\{1, 2, 3, 4\}$.
$$
\begin{tabular}{ c c | c | c |  c | c |}
\multicolumn{6}{c}{$\qquad B$}\\
\multicolumn{2}{c}{}&\multicolumn{1}{c}{$1$} &\multicolumn{1}{c}{$2$}&\multicolumn{1}{c}{$3$}&\multicolumn{1}{c}{$4$}\\
\cline{3-6}
&$1$&  $0.1$ & $0$ & $0$ & $0$  \\\cline{3-6}
A&$ 2$& $0.1$ & $0.2$ & $0$ & $0$  \\\cline{3-6}
&$ 3$& $0$ & $0$ & $0.1$ & $0.1$  \\\cline{3-6}
 &$4$& $0$ & $0$ & $0.2$ & $0.2$  \\\cline{3-6}
\end{tabular}
$$
The graph associated with this distribution is given in Figure~\ref{fig3}. This graph is disconnected. The common data of $A, B$ is one bit, determined by whether $A$ and $B$ are both in $\{1,2\}$ or in $\{3,4\}$.

\begin{figure}
\centering
\includegraphics[scale=0.8]{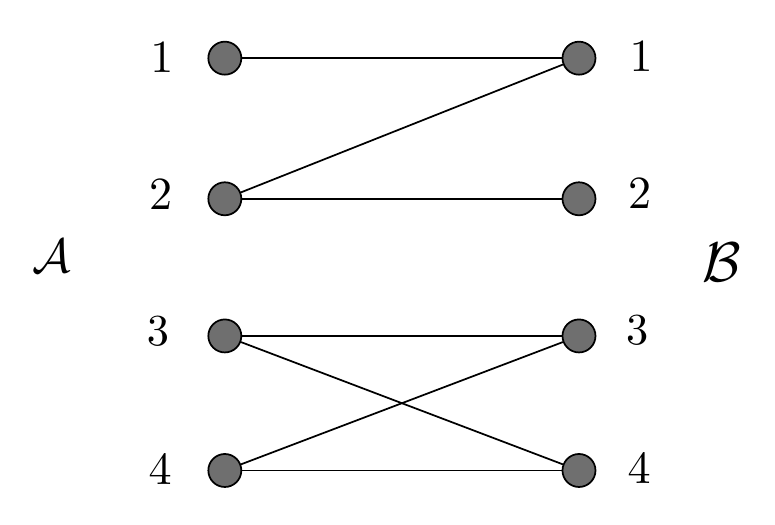}
\caption{The graph associated with probability distribution given in Example \ref{ex9}. This graph has two non-singleton connected components, so $A$ and  $B$ have common data.}
\label{fig3}
\end{figure}

\end{example}

Let $c\in \mathcal C$ be a connected component of $\mathcal G$. Then $p(ab|c)$, the distribution of $A, B$ conditioned on $C=c$, does not have common data. This is because the bipartite graph associated to this conditional distribution is nothing but the $c$-th connected component of $\mathcal G$, which by definition is connected. Denoting the maximal correlation of this conditional distribution by $\rho(A, B|C=c)$ we find that $\rho(A, B|C=c)<1$.

\begin{definition}[Conditional maximal correlation \cite{BeigiTse}]\label{def5}
Let $p(abc)$ be a tripartite distribution. We define
$$\rho(A, B|C) := \max_{c: \, p(c) > 0} \rho(A, B|C=c),$$
where $\rho(A, B|C=c)$ is the maximal correlation of the conditional bipartite distribution $p(ab|c)$.
\end{definition}

With this definition, for all bipartite distributions $p(ab)$, if we define $C$ to be the common part of $A$ and $B$ as in Lemma~\ref{lem:bipartite-graph} then
\begin{align}\label{eq:rho-abc-l-1}
\rho(A, B|C)<1.
\end{align}

\subsection{Common data of a distributed SV source}

Given a distributed SV source specified by distributions $p_s(ab)$, our goal is to determine whether a common random bit can be extracted from this source or not. Suppose that for some $s\in \mathcal S$, the maximal correlation of $p_s(ab)$, which we denote by $\rho_s(A, B)$, is less than $1$. Then, by Theorem~\ref{thm:Witsenhausen} common randomness extraction is impossible because the adversary can in all time steps choose $s_i=s$ to obtain an i.i.d.\ source. So we may assume that $\rho_s(A, B)=1$ for all $s$.

Let $\mathcal G_s$ be the bipartite graph associated to the bipartite distribution $p_s(ab)$. By the above observation we let $\rho_s(A, B)=1$ and then $\mathcal G_s$ has at least two non-singleton connected components. We claim that even the graph $\bigcup_s \mathcal G_s$, obtained by the union of edges of individual graphs $\mathcal G_s$, should also have at least two non-singleton components. To see this, assume that the adversary in each time step chooses $s_i\in \mathcal S$ uniformly at random and independent of the past. Then we obtain an i.i.d.\ source with distribution
$$q(ab)= \frac{1}{|\mathcal S|} \sum_s p_s(ab).$$
Then common randomness can be extracted from this i.i.d.\ source, only if the bipartite graph associated to $q(ab)$ is disconnected. It is easy to verify that this bipartite graph is nothing but $\bigcup_s \mathcal G_s$. So without loss of generality we may assume that $\bigcup_s \mathcal G_s$ has at least two non-singleton connected components.

The following lemma summarizes the above discussion.

\begin{lemma}\label{lem:SV-common}
For a distributed SV source we let $\mathcal G_s$ be the bipartite graph associated to $p_s(ab)$, and define $\bar{\mathcal G}=\bigcup_s \mathcal G_s$. Then a common bit can be extracted from the distributed SV source only if $\bar{\mathcal G}$ has at least two non-singleton connected components. Moreover, letting $C$ be the random variable corresponding to the connected components of $\bar{\mathcal G}$, then $C$ can be computed by Alice and Bob separately.
\end{lemma}

\begin{definition}\label{def:SV-common}
The random variable $C$ defined in Lemma~\ref{lem:SV-common} is called the common data of the distributed SV source.
\end{definition}

\begin{example}
\label{ex10}
Consider the following two joint distribution on $A$ and $B$. The graph corresponding to both of these distributions has three connected components. But if we superimpose these two distributions over each other (by choosing each with probability half), the graph of the resulting distribution has only two connected components.
$$
\begin{tabular}{ c c | c | c |  c | c |}
\multicolumn{6}{c}{$\qquad B$}\\
\multicolumn{2}{c}{}&\multicolumn{1}{c}{$1$} &\multicolumn{1}{c}{$2$}&\multicolumn{1}{c}{$3$}&\multicolumn{1}{c}{$4$}\\
\cline{3-6}
&$1$&  $0.1$ & $0$ & $0$ & $0$  \\\cline{3-6}
A&$ 2$& $0$ & $0.2$ & $0$ & $0$  \\\cline{3-6}
&$ 3$& $0$ & $0$ & $0.1$ & $0.1$  \\\cline{3-6}
 &$4$& $0$ & $0$ & $0.3$ & $0.2$  \\\cline{3-6}
\end{tabular}\qquad\qquad
\begin{tabular}{ c c | c | c |  c | c |}
\multicolumn{6}{c}{$\qquad B$}\\
\multicolumn{2}{c}{}&\multicolumn{1}{c}{$1$} &\multicolumn{1}{c}{$2$}&\multicolumn{1}{c}{$3$}&\multicolumn{1}{c}{$4$}\\
\cline{3-6}
&$1$&  $0.2$ & $0$ & $0$ & $0$  \\\cline{3-6}
A&$ 2$& $0.1$ & $0.1$ & $0$ & $0$  \\\cline{3-6}
&$ 3$& $0$ & $0$ & $0.3$ & $0$  \\\cline{3-6}
 &$4$& $0$ & $0$ & $0$ & $0.3$  \\\cline{3-6}
\end{tabular}
$$
\end{example}

\subsection{Common random bit extraction from distributed SV sources}

We now have all the required tools to state and prove our main result about common randomness extraction from distributed SV sources.
\begin{theorem}\label{thm:main-distributed}
Consider a distributed SV source (as in Definition~\ref{def:dist-SV}) with corresponding sets $\mathcal S$, $\mathcal A$, and $\mathcal B$ and corresponding distributions $p_s(ab)$.  Let $C$ be the common data of the distributed SV source (as in Definition~\ref{def:SV-common}).
Let $p_s(abc)$ denote the induced joint distribution of $A$, $B$, and $C$.                                                                                                                                                                                                                                                                                                                                                                                                                                                                                                                                                                                                                                                                                                                                                                                                                                                                                                                                                                                                                                                                                                                                                  Suppose that there is no non-zero function $\psi:\mathcal C\rightarrow \mathbb R$ such that $\E_{(s)}[\psi(C)]=0$ for all $s$.
Then common randomness cannot be extracted from this distributed SV source.
\end{theorem}

An algorithm to extract common random bits is to focus on the common part $C$ that can be computed by both Alice and Bob. Indeed $C$ itself can be thought of as a generalized SV source. If deterministic randomness extraction from $C$ is possible, then Alice and Bob can obtain a common random bit by individually applying the randomness extraction protocol. Comparing with Theorems~\ref{thm:martingale}~and~\ref{thm:necessary-condition}, and assuming $p_s(c)>0$ for all $s, c$, the above theorem states that a common random bit can be extracted if and only if deterministic randomness extraction from $C$ is possible.

The proof of this theorem is essentially obtained by combining the ideas developed in the proofs of Theorems~\ref{thm:necessary-condition}~and~\ref{thm:Witsenhausen}. We present a detailed proof in Appendix~\ref{appendixD}.

\section{Future Work}\label{sec:future-work}

In this paper we completely characterized the randomness extraction problem for non-degenerate cases. A future work could be to solve this problem for the degenerate cases. 
In the degenerate cases, for generalized non-distributed sources Corollary~\ref{cor:necessary} gives a mildly stronger necessary condition than
Theorem~\ref{thm:necessary-condition}, but there is still a gap between this necessary condition 
and the sufficient condition of Theorem~\ref{thm:martingale}.

We note that our randomness extractor in Theorem~\ref{thm:martingale} extracts a bit
whose bias is inverse polynomially small in the length of the source sequence. 
It is interesting to see if this extractor could be improved to yield a bit with an exponentially small bias.
Furthermore, if we want to produce more than one bit of randomness, the tradeoff between the number of produced random bits and their quality is open.

Another interesting problem is to look at efficient adversaries, similar to the work of~\cite{ACMPS}. Our proofs only show existence of inefficient adversaries.

Another way to restrict the adversary is to put limitations on the number of times the adversary can choose a strategy $s\in\mathcal{S}$, i.e.\ there can be a cost associated to each strategy~$s$. 

A different type of limitation can be on the adversary's knowledge about the sequence generated so far.
More specifically, the adversary might have  \emph{noisy or partial} access to the previous outcomes in the sequence (these sources are called ``active sources" \cite{Sahai2}). These sources model adversaries with limited memory. Space bounded sources have been studied in \cite{Kamp,Vazirani87}.

%%%%%%%%%%%%%%%%%%%%%%%%%%%%%%%%%%%%%%%%%
\appendix

\section{Exact bias of deterministic extractors for the SV source}\label{appendixC}

SV sources were originally defined in the binary case~\cite{SanthaVazirani}. Such a source is specified by two distributions, i.e., $\mathcal S=\{0,1\}$, over $\mathcal C=\{0,1\}$ with
$$p_0(0)=\delta,\qquad \text{ and } \qquad p_1(0)=1-\delta,$$
where $0<\delta<1/2$. This is proved in~\cite{SanthaVazirani} and can also be concluded from Theorem~\ref{thm:necessary-condition} that randomness extraction from this SV source is impossible. Our goal in this appendix is to exactly characterize the set $\cup_{n}\Phi_n$ for this source where $\Phi_n$ is defined in the proof of Theorem~\ref{thm:necessary-condition}.

Let us describe our problem here more precisely.
\begin{definition}
 Fix an algorithm for extracting randomness from the binary SV source with parameter $\delta$. Let $\alpha$ be the minimum of the probability of the extracted bit being 0, where the minimum is taken over all adversary's strategies. Similarly let $\beta$ be the maximum of this probability over all strategies of the adversary.  
We call $(\alpha, \beta)$ the pair associated with the extractor.
Define
$H_\delta$ be the set of all such pairs $(\alpha, \beta)$ over all possible extractors.
\end{definition}
Our goal is to determine the set $H_\delta$.

To state the result we need some notation.
\begin{definition}
 Fix $0<\delta<1$.
For $x_1, x_2, \ldots x_n \in \{0,1\}$, define
$$(0.x_1 x_2 \ldots x_n)_\delta = \sum_{i=1}^n\, x_i (1-\delta)^i \big(\frac{\delta}{1-\delta}\big)^{s_x(i)},$$
where
$$s_x(i) = \sum_{j=1}^{i-1}x_j.$$
\end{definition}
Observe that when $\delta=1/2$, we get the standard binary expansion.

\begin{definition}
For two pairs $(\alpha_1, \beta_1)$ and $(\alpha_2, \beta_2)$ of real numbers we say that $(\alpha_1, \beta_1)$ \emph{dominates} $(\alpha_2, \beta_2)$ if $\alpha_1 \le \alpha_2$ and $\beta_1 \ge \beta_2$.
\end{definition}

The set $H_{\delta}$ can be characterized using the following proposition that implicitly appears in the conference version of \cite{SanthaVazirani} (at the beginning of their sketch of proof of their Theorem~6).

As mentioned in the introduction, a deterministic extractor has a corresponding  depth-$n$ binary tree, with leaves marked by either 0 or 1.

\begin{proposition}[\cite{SanthaVazirani}]\label{SV-thm-conf}  Assume that the depth-$n$ binary tree associated with the deterministic extractor has exactly $x$ leaves that are marked with bit 0. Let $x = (x_1 \ldots x_n)_2$ be the binary expansion of $x$.
Then the maximum probability $y$ that the extracted bit is 0 is at least $(0.x_1 \ldots x_n)_{\delta}$,
and the equality occurs when the $x$ leaves of value 0 form a left prefix of all leaves,
i.e., they appear consecutively from the leftmost leaf towards right
(or in other words the extractor assigns 0 to the sequence $(y_1, \cdots, y_n)$ iff $(y_1, \cdots, y_n)_2 < x.$)
\end{proposition}

The following Corollary implies Figure~\ref{Fig1} for the binary SV source with $\delta = 1/3$.
\begin{corollary}\label{thm:16} Let
$$F_\delta:= \big\{((0.x_1 \ldots x_n)_{1 - \delta}, (0.x_1 \ldots x_n)_{\delta}):\, \forall n, \forall x_1, \ldots x_n \in \{0,1\}\big\}\cup \big\{(1,1)\big\}.$$
Then $F_\delta\subseteq H_\delta$.
Furthermore, any $(\alpha, \beta)\in H_\delta$ is dominated by a pair in $F_\delta$.
\end{corollary}
\begin{proof}
By symmetry, the maximum probability that the extracted bit be 1 is minimized when all leaves with value 1 form a left prefix,
hence when the leaves with value 0 form a right prefix.
In other words (and again by symmetry), the minimum probability that the extracted bit be 0 is maximized when all leaves with value 0
form a left prefix.
Thus, both minimum of $y$ and maximum of $x$ occur when all 0-leaves
form a left prefix.
Observe that the pair $(x,y)$ associated to the tree having 0-leaves as a left prefix is
$((0.x_1 \ldots x_n)_{1 - \delta}, (0.x_1 \ldots x_n)_{\delta})$.
\end{proof}

Santha and Vazirani argue that Proposition \ref{SV-thm-conf} follows from inequality \eqref{eqn:SV-do-not-prove}, which is not proved in their paper. Lemma \ref{lemma:basedelta} below gives a proof for the inequality and hence the proposition.

\begin{lemma}
\label{lemma:basedelta}
 If \begin{equation}\label{eq:a+b=c}(0.x_1 \ldots x_n)_{1/2} + (0. y_1 \ldots y_n)_{1/2} = (0.z_1 \ldots z_{n})_{1/2}\end{equation}
and \begin{equation}\label{eq:a>b}(0.x_1 \ldots x_n)_{1/2} \ge (0.y_1 \ldots y_n)_{1/2},\end{equation}
then \begin{equation}\label{eq:result}(0.x_1 \ldots x_n)_\delta + \frac{\delta}{1-\delta}(0.y_1 \ldots y_n)_\delta \ge (0.z_1 \ldots z_n)_{\delta}.\end{equation}
\end{lemma}

\begin{remark} This lemma in particular shows that if $x = (x_1 \ldots x_n)_2, y = (y_1 \ldots y_n)_2, z = (z_1 \ldots z_{n+1})_2, x + y = z, x \ge y$,
then $(0.0x_1 \ldots x_n)_{1/2} + (0. 0y_1 \ldots y_n)_{1/2} = (0.z_1 \ldots z_{n+1})_{1/2}$, and hence  \begin{align}(1-\delta) (0.x_1 \ldots x_n)_\delta + \delta (0.y_1 \ldots y_n)_\delta \ge (0.z_1 \ldots z_{n+1})_\delta. \label{eqn:SV-do-not-prove}\end{align} This latter equation proves the induction step in the proof of
\cite{SanthaVazirani}.
\end{remark}

\begin{proof}[Proof of Lemma~\ref{lemma:basedelta}]
First, we show that without loss of generality we may assume $x_i \ge y_i$ for all $i$.
For the first $i$ for which $x_i \ne y_i$, we have $x_i > y_i$.
Consider the first $i$ for which $x_i < y_i$.
If for this $i$, we swap $x_j$ with $y_j$ for all $j \ge i$, we still have Equation (\ref{eq:a+b=c})
but $(0.x_1 \ldots x_n)_\delta + \frac{\delta}{1-\delta}(0.y_1 \ldots y_n)_\delta$  decreases  by
$$(1-\delta)^{i-1}((\frac{\delta}{1-\delta})^{s_x(i)} - (\frac{\delta}{1-\delta})^{s_y(i) + 1})((0.x_i \ldots x_n)_\delta - (0.y_i \ldots y_n)_\delta)$$
which is nonnegative because $s_x(i) \ge s_y(i) + 1$ and $(0.x_i \ldots x_n)_\delta \le (0.y_i \ldots y_n)_\delta$.
We can successively do these swaps until $x_i \ge y_i$ for all $i$.

Now we prove the lemma by induction on $n$.
Assume for the sake of contradiction that Inequality (\ref{eq:result}) does not hold.

If $z_1 = 0$, then $x_1 = y_1 = z_1 = 0$. Then we can remove $x_1, y_1, z_1$, decrease $n$ by 1, and prove the lemma using the induction hypothesis.  Thus, assume that $z_1=1$.
Now we can partition the indices $\{1, \ldots, n\}$ into blocks
such that in the addition of $(0.x_1 \ldots x_n)_{1/2}$ with $(0.y_1 \ldots y_n)_{1/2}$,
no carry is passed from one block to the next block, but within each block there is always a passed carry.
Consider the leftmost block that begins from index $1$ and ends at index $m$.

If $m = 1$, then we should have $x_1 = 1, y_1 = 0, z_1 = 1$.
If we change $x_1$ and $z_1$ to 0, then we still have Equation (\ref{eq:a+b=c}).
Also, Inequality (\ref{eq:a>b}) holds because $x_i \ge y_i$ for $i \ge 2$.
Furthermore,
both $(0.x_1 \ldots x_n)_{\delta}$ and $(0.z_1 \ldots z_{n})_{\delta}$ are decreased by $1-\delta$ and then multiplied by $\delta/(1-\delta) \le 1$,
while $(0.y_1 \ldots y_n)_{\delta}$ does not change.
Therefore, Inequality (\ref{eq:result}) holds, if and only if it holds after changing $x_1$ and $z_1$ to 0.
Since now $z_1 = 0$, we can use the induction hypothesis.

If $m > 1$, then we should have $x_1 = 0, x_2 = x_3 = \ldots = x_m = 1, y_1 = 0, z_1 = 1, y_2 = z_2, \ldots, y_{m-1} = z_{m-1}, y_m = 1, z_m = 0$.
Let $i$ be an index $\in [2,m-1]$ such that $y_i = 0$ (if such an $i$ exists.)
If we change $y_i$ and $z_i$ both to 1, then Equation (\ref{eq:a+b=c}) holds. Inequality (\ref{eq:a>b}) also holds since $x_i
\geq y_i$ for all $i$. Furthermore, $\frac{\delta}{1-\delta}(0.y_1 \ldots y_n)_\delta - (0.z_1 \ldots z_n)_\delta$ decreases by
\begin{align*}&\frac{\delta}{1-\delta}(1-\delta)^{i-1}(\frac{\delta}{1-\delta})^{s_y(i)}(1-\frac{\delta}{1-\delta}) ((0.y_{i} \ldots y_n)_{\delta})\\&\quad
-(1-\delta)^{i-1}(\frac{\delta}{1-\delta})^{s_z(i)}(1-\frac{\delta}{1-\delta}) ((0.z_{i} \ldots z_n)_{\delta})
\\&=\frac{\delta}{1-\delta}(1-\delta)^{i-1}(\frac{\delta}{1-\delta})^{s_y(i)}(1-\frac{\delta}{1-\delta}) ((0.y_{i} \ldots y_n)_{\delta}-(0.z_{i} \ldots z_n)_{\delta})
\end{align*}
because $s_y(i) = s_z(i) - 1$.
This decrease is nonnegative because $(0.y_{i+1} \ldots y_n)_{\delta}> (0.z_{i+1} \ldots y_n)_{\delta}$
(since $y_m > z_m$).
So, to prove the claim, without loss of generality we can assume $x_2 = \ldots = x_m = y_2 = \ldots = y_m = z_1 = \ldots = z_{m-1} = 1$. But now if we make the values of $x_i, y_i, z_i$ to 0 for all $i \in [1,m]$,
we still have Equality (\ref{eq:a+b=c}) and Inequality (\ref{eq:a>b}),
while this does not change the difference of the two sides of Inequality (\ref{eq:result}).
Given $z_1 = 0$, we can use the induction hypothesis as above.
\end{proof}

%***********************************************
\section{Another proof of Theorem~\ref{thm:necessary-condition}}\label{app:proof1-necessary}
%(\emph{Sketch of proof of Theorem~\ref{thm:necessary-condition}}):
Consider the set of points $\{p_s(\cdot): s \in \mathcal S\}$ in the probability simplex.
Then, by assumption there is a point $q(\cdot)$ in the interior of the convex hull of these points.
Fix a deterministic extractor specified by a subset
$\mathcal I\subseteq {\mathcal C}^n$, i.e., if the observed $c^n$ is in $\mathcal I$ then the extracted bit is $0$, and otherwise it is $1$.
Consider the probability distribution $q^n(\cdot)$ on $\mathcal{C}^n$ that is the i.i.d. repetition of $q(\cdot)$.
Without loss of generality, assume that $q^n(\mathcal I)\ge 1/2$.
Let $\mathcal I_0\subseteq \mathcal I$ be a minimal subset such that $q^n(\mathcal I_0)\geq 1/2$. That is, let $\mathcal I_0\subseteq \mathcal I$ be such that $q^n(\mathcal I_0)\geq 1/2$ and no proper subset of $\mathcal I_0$ has this property.  Observe that for any $c^n\in \mathcal{C}^n$ we have $q^n(c^n)\leq 2^{-\Theta(n)}$. Therefore, by the definition of $\mathcal I_0$ we have $q^n(\mathcal I_0)=1/2+2^{-\Theta(n)}$.

Let $\widetilde p(\cdot)$ be a tweak of the distribution $q^n(\cdot)$ obtained as follows. Let $\epsilon>0$ be small constant, and define $\widetilde p(c^n) = (1+\epsilon)q^n(c^n)$ for $c^n\in \mathcal I_0$; also for $c^n\notin \mathcal I_0$ define $\widetilde p(c^n)=(1 - \epsilon -2^{-\Theta(n)}) q^n(c^n)$ to make $\widetilde p(\cdot)$ a probability distribution.

We claim that $\widetilde p(\cdot)$ is in the class of distributions associated with the generalized SV source, i.e., the adversary can choose a strategy to generate this distribution. Assuming this claim, observe that the probability that the extracted bit is $0$ would be equal to
$$\widetilde p(\mathcal I) \geq \widetilde p(\mathcal I_0)=(1+\epsilon)q^n(\mathcal I_0) \geq (1+\epsilon)/2.$$
Thus the adversary can force the bias of the extracted bit to be at least $\epsilon$. This would finish the proof.

What remains to show is that $\widetilde p(\cdot)$ can be generated by the adversary.
Observe that for any $\mathcal J\subseteq \mathcal C^n$ we have
$$\widetilde p(\mathcal J) = (1 + O(\epsilon))q^n(\mathcal J).$$
In particular, for any $c_1, \ldots, c_i$, we have
$$\widetilde p(C_1=c_1, \dots, C_i=c_i)= (1 + O(\epsilon)) \prod_{j=1}^i q({c_j}).$$
Therefore,
$$\widetilde p(C_i=c_i| C_1=c_1, \dots, C_{i-1})=\frac{1+O(\epsilon)}{1+O(\epsilon)} q(c_i) = (1 + O(\epsilon)) q(c_i).$$
Since $q(\cdot)$ is in the interior of the convex hull of $\{p_s(\cdot): s \in \mathcal S\}$, then for sufficiently small $\epsilon>0$,
any probability distribution of the form $((1 + O(\epsilon))q(\cdot)$
is in this convex hull too.
Thus, by definition $\widetilde p(\cdot)$ can be produced by the adversary.

%**************************************
\section{Proof of Lemma \ref{lemma1}}\label{app:lemma1}

We use the notation developed before the statement of Lemma~\ref{lemma1}. We assume that $K_1=\Gamma_n(A^n, R_1)$ and $K_2=\Lambda_n(B^n, R_2)$ are $\epsilon$-close to the uniform distribution over $\B$ and that $\Pr[K_1\neq K_2]<\epsilon$.
Define
$$K'_1=\Gamma'_n(A^n)=\text{argmax}_{k}\,\Pr[K_1=k|A^n],$$
and
$$K'_2=\Lambda'_n(B^n)=\text{argmax}_{k}\, \Pr[K_2=k|B^n].$$
Then $K'_1$ and $K'_2$ are (deterministic) functions of $A^n$ and $B^n$ respectively.
We claim that
$$\Pr[K'_1\neq K'_2]\leq 3\epsilon,$$
and
$$ \big| \Pr[K'_1=0]-\frac 12\big|, \big| \Pr[K'_2=0]-\frac 12\big|\leq 2\epsilon.$$
Proving these inequalities would complete the proof.

Observe that for every $A^n=a^n$ we have
$$\Pr[K'_1\neq K_1|A^n=a^n]=\min\{\Pr[K_1=0|a^n], \Pr[K_1=1|a^n]\}.$$
On the other hand,
\begin{align*}
\Pr[K_2\neq K_1|a^n]&=\Pr[K_2=1|a^n]\Pr[K_1=0|a^n]+\Pr[K_2=0|a^n]\Pr[K_1=1|a^n]\\
&\geq \min\{\Pr[K_1=0|a^n], \Pr[K_1=1|a^n]\}
\\&=\Pr[K'_1\neq K_1|A^n=a^n].
\end{align*}
As a result, we have
$$\Pr[K'_1\neq K_1]\leq \Pr[K_2\neq K_1]\leq \epsilon,$$
which gives
$$\big| \Pr[K'_1=0]-\frac 12\big|\leq 2\epsilon.$$
This inequality for $K'_2$ is proved similarly.

Next we have
\begin{align*}
\Pr[K'_1\neq K'_2]&\leq \Pr[K'_1\neq K_1]+\Pr[K_1\neq K_2]+\Pr[K_2\neq K'_2]
\\&\leq \epsilon+\epsilon+\epsilon.
\end{align*}

%***********************************************************************
\section{Proof of Theorem~\ref{thm:main-distributed} }\label{appendixD}

First we show that it suffices to prove Theorem~\ref{thm:main-distributed} in the following special case.

\begin{lemma}\label{lem:special-case}
If Theorem~\ref{thm:main-distributed} holds in the special case where distributions $p_s(a, b)$ satisfy
\begin{align}\label{eq:psab-positive}
p_s(a), p_s(b)>0,\qquad\qquad \forall s, a, b.
\end{align}
and 
\begin{align}\label{eq:SV-conditional-rho}
\rho(A, B|CS):=\max_s \rho_s(A, B|C)<1,
\end{align}
where $\rho_s(A, B|C)$ denotes the conditional maximal correlation of $A$ and $B$ given $C$ with respect to the distribution $p_s(abc)$,
 then the theorem holds in general.

\end{lemma}

\begin{proof}
For any $s\in \mathcal S$ define $p'_s(a, b)$ as follows: fix some (small) $\tau\in (0,1)$, and define 
\begin{align}\label{eq:def-qs}
p'_s(a, b):=(1-\tau)p_s(a, b) + \frac{\tau}{|\mathcal S|}\sum_{s'\in \mathcal S} p_{s'}(a, b).
\end{align}
Observe that $p'_s(\cdot)$ is some perturbed variant of $p_s(\cdot)$. Also note that $p'_s(\cdot)$ is in the convex hull of distributions $p_{s'}(\cdot)$ for different values of $s'$. Thus in each step, the adversary can enforce that the pair $(a, b)$ generated by the source has distribution $p'_s(a, b)$ via a randomized strategy. As a result, it suffices to show the impossibility of common randomness extraction from the distributed SV source with distributions $p'_s(\cdot)$ instead of $p_s(\cdot)$.
Now we only need to show that $p'_s(\cdot)$ satisfies~\eqref{eq:psab-positive} and~\eqref{eq:SV-conditional-rho} as well as the assumption of Theorem~\ref{thm:main-distributed}.

First, by the definitions of $p'_s(\cdot)$ given in~\eqref{eq:def-qs} the support of $p'_s(\cdot)$ does not depend on $s$. Therefore, without loss of generality we may assume that for any $a\in \mathcal A$ we have $p'_s(a)>0$. Similarly, we may assume that for any $b\in \mathcal B$ we have $p'_s(b)>0$. 

Second, it is not hard to see that the graph $\mathcal G'_s$ associated with distributions $p'_s(\cdot)$ is the same for all $s$ and is equal to the graph $\bar{\mathcal G}=\cup_s \mathcal G_s$ associated with the original distributions $p_s(\cdot)$. This in particular implies that $\bar{\mathcal G'}=\bar{\mathcal G}$ and the common part $C$ remains the same. Moreover, for any $s, c$, we have $\rho'_s(A, B| C=c)<1$ for distribution $p'_s(\cdot)$ because the connected components of the graph
$\bar{\mathcal G}$ are nothing but elements of $\mathcal C$.

We finally verify that there is no non-zero $\psi:\mathcal C\rightarrow \mathbb R$ such that $\E'_{(s)}[\psi(C)]=0$ where the expectation is computed with respect to $p'_s(\cdot)$. Suppose such a function $\psi$ exists. Then we have 
\begin{align}\label{eq:234}
0=\E'_{(s)}[\psi] = (1-\tau)\E_{(s)}[\psi] + \frac{\tau}{|\mathcal S|}\sum_{s'} \E_{(s')}[\psi],
\end{align}
where $\E_{(s)}[\cdot]$ denotes expectation with respect to $p_s(\cdot)$.
Summing the above equations for all $s\in \mathcal S$, we find that 
$$\sum_s \E_{(s)}[\psi]=0,$$
and then using~\eqref{eq:234} again we obtain $\E_{(s)}[\psi]=0$. Therefore by the assumption of Theorem~\ref{thm:main-distributed} the function $\psi$ should be zero. 

\end{proof}

By the above lemma, from now on we assume that the distributions $p_s(\cdot)$ satisfy the extra assumptions~\eqref{eq:psab-positive} and~\eqref{eq:SV-conditional-rho}.

Suppose that common random bit extraction is possible.
By Lemma \ref{lemma1} we may assume that Alice and Bob's protocol is deterministic and is described by subsets $\mathcal I\subseteq \mathcal A^n$ and $\mathcal J\subseteq \mathcal B^n$. That is, Alice's output is $K_1=0$ if $a^n\in \mathcal I$ and Bob's output is $K_2=0$ if $b^n\in \mathcal J$.

Let us define
\begin{align*}
&\alpha(\mathcal I):=\max \Pr[A^n\in\mathcal{I}],\\
&\beta(\mathcal J):=\max \Pr[B^n\in\mathcal{J}],\\
&\gamma(\mathcal I, \mathcal J):=\min \Pr[A^n\in\mathcal{I}, \, B^n\in\mathcal{J}],
\end{align*}
where the maximizations  and the minimization are computed over all strategies of the adversary.
If common randomness extraction is possible, then there are $n$ and $\mathcal I\subseteq \mathcal A^n$ and $\mathcal J\subseteq \mathcal B^n$
such that
all the three numbers $\alpha(\mathcal I), \beta(\mathcal J)$ and $\gamma(\mathcal I, \mathcal J)$ are close to $1/2$.

For $n\geq 1$ let $\Phi_n$ be the set of triples $(\alpha(\mathcal I), \beta(\mathcal J), \gamma(\mathcal I, \mathcal J))$ for all subsets $\mathcal I\subseteq \mathcal A^n$ and $\mathcal J\subseteq \mathcal B^n$. We also define
$$\Phi_0=\big\{\mathbf{e}_0=(1, 1, 1), \mathbf{e}_1=(1, 0, 0), \mathbf{e}_2=(0, 1, 0), \mathbf{e}_3=(0, 0, 0)\big\}.$$
As discussed in the proof of Theorem~\ref{thm:Witsenhausen} the set $\Phi_0$ corresponds to deterministic strategies where the parties do not look at the source at all.
By the above discussion we need to show that $(1/2, 1/2, 1/2)$ is far from $\cup_n \Phi_n$.

By the same ideas as in the proofs of Theorems~\ref{thm:necessary-condition} and~\ref{thm:Witsenhausen} the sets $\Phi_n$ can be computed recursively.
For every $a,b$ and  $\mathcal I\subseteq \mathcal A^n$ and $\mathcal J\subseteq \mathcal B^n$ define
$$\mathcal I_{a}:=\{a_{[2:n]}:\, (a, a_{[2:n]})\in \mathcal I\},\qquad\qquad \mathcal J_{b}:=\{b_{[2:n]}:\, (b, b_{[2:n]})\in \mathcal J\}.$$
Then we have
\begin{align*}
\alpha(\mathcal I)&=\max_s \E_{(s)}[\alpha(\mathcal I_{A})]\\
\beta(\mathcal J)&=\max_s \E_{(s)}[ \beta(\mathcal J_{B})],\\
\gamma(\mathcal I, \mathcal J)&=\min_s \E_{(s)}[ \gamma(\mathcal I_{A}, \mathcal J_{B})].
\end{align*}
As a result, the sets $\Phi_n$ can be characterized recursively as follows. $\Phi_n$ is indeed the set of triples $(x, y, z)$ for which there are functions $X(a)=x_a$, $Y(b)=y_b$ and $Z(ab)=z_{ab}$ such that for all $(a, b)$ we have $(x_a, y_b, z_{ab})\in \Phi_{n-1}$ and that
\begin{align}\label{eq:xsX-2}
x = \max_s  \E_{(s)}[X], \qquad y=\max_s \E_{(s)}[Y], \qquad z=\min_s\E_{(s)}[Z].
\end{align}

We now prove that $\Phi_n$ for every $n$ is far from $(1/2, 1/2, 1/2)$.

\begin{theorem} \label{lem:f-ineq}
Let
$$0<\epsilon\leq \frac{\Delta' (1-\rho)}{1+\Delta'},$$
and $M\geq 24|\mathcal A||\mathcal B|/\Delta+2$ where
$\Delta$ and $\Delta'$ are two positive constants that are specified later (in Lemmas \ref{lem:Delta} and \ref{lem:Delta-prime}). Define
$$f(x,y,z)=M(x+y)-2(M+\epsilon)z+2xy -(1-\epsilon)(x^2+y^2).$$
Then with the assumption of Theorem~\ref{thm:main-distributed} and~\eqref{eq:psab-positive} and~\eqref{eq:SV-conditional-rho}, for all functions $X, Y, Z$ as above, we have
$$f( x,  y,  z)\geq \min_s \E_{(s)}[f(X, Y, Z)],$$
where $ x,  y,  z$ are defined in~\eqref{eq:xsX-2}.
\end{theorem}

Given this theorem we can finish the proof of Theorem~\ref{thm:main-distributed}. Observe that $f(\mathbf e_i)\geq 0$ for $0\leq i\leq 3$. Then by the above theorem and a simple induction, for  any $(x, y, z)\in \Phi_n$ we have $f(x, y, z)\geq 0$. We however have $f(1/2, 1/2, 1/2)= -\epsilon/2<0$. Then by the continuity of $f$, the point $(1/2, 1/2, 1/2)$ is far from $\Phi_n$ for any $n$.

The proof of Theorem \ref{lem:f-ineq} is the most technical part of this paper; its proof is given after stating some definitions and  lemmas.

\subsection{Some preliminary definitions and lemmas}

In this section, we let $\mathbb E_{(cs)}[\cdot]$ to be the expectation with respect to the conditional probability distribution $p_s(ab|c)$.

\vspace{.2in}
\noindent
\textbf{A characterization of conditional maximal correlation.}
 The assumption that $\rho=\rho(A; B|CS)<1$ implies the following for any fixed value of $s$:  take two arbitrary functions $X:\mathcal A\rightarrow \mathbb R$ and $Y: \mathcal{B}\rightarrow \mathbb R$ such that $$\mathbb E_{(cs)}[X]=\mathbb E_{(cs)}[Y]=0,\qquad \forall c,$$
where by the notation that we have set up before $\mathbb E_{(cs)}[X]=\sum_{ab}p_s(ab|c)x_a=\sum_{a}p_s(a|c)x_a$, and similarly $\E_{(cs)}[Y] = \sum_b p_s(b|c) y_b$.
Then we must have
\begin{align}\E_{(cs)}[XY] \le \rho \sqrt{\E_{(cs)}[X^2] \E_{(cs)}[Y^2]}\label{eqn:omid-rho-cond}\end{align}
for all $c$.
 Using the joint convexity of $f(x,y)=\sqrt{xy}$ we have that 
\begin{align}
\mathbb E_{(s)}[XY]&\leq   \sum_{c}p_s(c)\rho\sqrt{\E_{(cs)}[X^2] \E_{(cs)}[Y^2]}\nonumber
\\&\leq   \rho\sqrt{\big(\sum_{c}p_s(c)\E_{(cs)}[X^2]\big)\big(\sum_{c}p_s(c)\E_{(cs)}[Y^2]\big)}\nonumber
\\&
= \rho\sqrt{\mathbb E_{(s)}[X^2]\mathbb E_{(s)}[Y^2]}.\label{eq:rho-conditional-zero-mean}
\end{align}
On the other hand, equation \eqref{eq:rho-conditional-zero-mean} implies equation \eqref{eqn:omid-rho-cond} by choosing $X$ and $Y$ to be zero whenever $C$ is not equal to some given $c$. Therefore \eqref{eq:rho-conditional-zero-mean} is a complete characterization of the conditional maximal correlation.

\vspace{.2in}\noindent
\textbf{Definitions of $\mathcal L_{A}$, $\mathcal L_{B}$, $\mathcal L_{A}^{\perp}$ and $\mathcal L_{B}^{\perp}$.}
Let $\mathcal L_{A}$ be the linear space of functions $X:\mathcal A \rightarrow \mathbb R$ such that $\E_{(s)}[X]$ is independent of $s$, i.e.,
\begin{align}\label{eqn:am25}
\mathcal L_{A}:=\{X:\mathcal A\rightarrow \mathbb R:\, \E_{(s)}[X]=\E_{(s')}[X],~ \forall s, s'\}.
\end{align}

Let $\mathcal L_{A}^{\perp}$ be the orthogonal complement of $\mathcal L_{A}$ with respect to the inner product $\langle \cdot, \cdot\rangle_*$, which is the inner product with respect to the uniform distribution, i.e.,
$$\mathcal L_A^\perp:=\{X:\mathcal A\rightarrow \mathbb R:\, \langle X, X'\rangle_*=0,~ \forall X'\in \mathcal L_A\}.$$
We define $\mathcal L_{B}$ and $\mathcal L_{B}^{\perp}$ similarly.

\begin{lemma}\label{lem:Delta} There is $\Delta>0$ such that for all vectors $X\in \mathcal L_{A}^{\perp}$ and $Y\in \mathcal L_{B}^{\perp}$ we have
$$\max_{s, s'} \E_{(s)}[X]-\E_{(s')}[X] \geq \Delta\|X\|_*,$$
and
$$\max_{s, s'} \E_{(s)}[Y]-\E_{(s')}[Y] \geq \Delta\|Y\|_*.$$

\end{lemma}

\begin{proof}
It suffices to show that
\begin{align*}
\max_{s, s'} \E_{(s)}[X]-\E_{(s')}[X] >0,
\end{align*}
for any $X\in \mathcal L_A^\perp$ with $\|X\|_*=1$.
The proof then follows from the compactness of the unit ball in $\mathcal L_{A}^{\perp}$. To show the above inequality note that the left hand side is always non-negative, and that it is zero if and only if $X\in \mathcal L_A$. But since $0\neq X\in \mathcal L_A^\perp$, it cannot be in $\mathcal L_A$. We are done.
\end{proof}

\vspace{.2in}
\noindent
\textbf{Definitions of $\mathcal L'_{A}$ and $\mathcal L'_{B}$.}
In Theorem~\ref{thm:main-distributed} we assume that there is no non-constant $\psi:\mathcal C\rightarrow \mathbb R$ such that $\E_{(s)}[\psi]$ is independent of $s$. To state this property in terms of our notations, let us define $\mathcal K_A$ be the set of functions $U: \mathcal A\rightarrow \mathbb R$ such that $U$ is determined by $C$, i.e.,
$$\mathcal K_A:=\{U: \mathcal A\rightarrow \mathbb R: \, U(a)=U(a'),~ \forall a, a' \text{ s.t. } C(a)=C(a')\}.$$
With abuse of notation for a function $U\in \mathcal K_A$ we may use $U(c)$ since $U$ is indeed a function of $C$.

Then the assumption of Theorem~\ref{thm:main-distributed} equivalently means that $\mathcal L_{A}\cap \mathcal K_A$ contains only constant functions, i.e.,
$$\mathcal L_A\cap \mathcal K_A=\{r\mathbf 1_A:\, r\in \mathbb R\}.$$

Let us define
\begin{align}
\mathcal L'_{A}:=\mathcal L_{A}\cap(\mathbf 1_{A})^{\perp},
\label{eqeq:aa1}
\end{align}
where $(\mathbf 1_{A})^{\perp}$ is computed with respect to the inner product $\langle \cdot, \cdot\rangle_*$.
Then the above condition implies that
\begin{align}
\mathcal L'_{A}\cap \mathcal K_A=\{0\}
\label{eqn:am5}.
\end{align}
We similarly define $\mathcal K_B$ and $\mathcal L'_{B}$ and have
$\mathcal L'_{B}\cap \mathcal K_B=\{0\}$.

\vspace{.2in}
\noindent
\textbf{Definitions of $\mathcal K_A^{\perp_s}$ and $\mathcal K_B^{\perp_s}$.}
Let $\mathcal K_A^{\perp_s}$ and $\mathcal K_B^{\perp_s}$ be the orthogonal complements of $\mathcal K_A$ and $\mathcal K_B$ respectively, with respect to the inner product $\langle \cdot, \cdot\rangle_{(s)}$.
We define $$\mathcal K_A^{\perp_s} = \{U': \mathcal A \rightarrow \mathbb R: \E_{(s)} [U' X] = 0, \ \forall X \in \mathcal K_A\},$$
and similarly we define $\mathcal K_B^{\perp_s}.$
Note that $\langle \cdot, \cdot \rangle_{(s)}$ is indeed an inner product because of assumption~\eqref{eq:psab-positive}. Then the above orthogonal complement is well-defined. 
Observe that
$$\mathcal K_A^{\perp_s}=\{U':\mathcal A\rightarrow \mathbb R:\, \E_{(cs)}[U']=0, ~\forall c \}.$$

We can write any function $X: \mathcal A \rightarrow \mathbb R$ as $X = U + U'$,
where $U \in \mathcal K_A$ and $U' \in \mathcal K_A^{\perp_s}$. Indeed, let
$$ U =
\E_{(C(a) s)} [X].$$
Then we have
\color{black}
$$\E_{(cs)}[U']=\E_{(cs)}[X-U]=\E_{(cs)}[X]-U(c)=0.$$
Therefore, by definition $U'\in \mathcal K_A^{\perp_s}$.

\begin{lemma} \label{lem:Delta-prime}
There is $\Delta'>0$ such that for any $X\in \mathcal L'_{A}$, $Y\in \mathcal L'_{B}$ and  $s\in \mathcal S$ we have
$$\|U'\|_{(s)}\geq \Delta'\|U\|_{(s)}, \qquad \quad \|V'\|_{(s)}\geq \Delta'\|V\|_{(s)},$$
where $U\in \mathcal K_A$ and $U'\in \mathcal K_A^{\perp_s}$ are such that $X=U + U'$. Functions $ V\in \mathcal K_B$ and $V'\in \mathcal K_B^{\perp_s}$ are defined similarly.

\end{lemma}

\begin{proof} Without loss of generality, we can restrict to $X\in \mathcal L'_{A}$ where $\|X\|_{(s)}=1$.
Using~\eqref{eqn:am5} we have $U'\neq 0$ for any such $X$. Thus $\|U\|_{(s)}/\|U'\|_{(s)}$ is well defined and continuous as a function on the unit sphere of $\mathcal L'_A$. Therefore, it achieves its maximum. Let $M_s<\infty$ be the maximum of $\|U\|_{(s)}/\|U'\|_{(s)}$ and $\|V\|_{(s)}/\|V'\|_{(s)}$ over the unit balls of $\mathcal L'_{A}$ and $\mathcal L'_{B}$. Then the the choice of $\Delta'=\min_s(1/M_s)$ works.
\end{proof}

Now we have all the required tools to prove Theorem~\ref{lem:f-ineq}.

\subsection{Proof of Theorem~\ref{lem:f-ineq}}
First note that $f(x, y, z)$ is monotonically increasing in its first and second arguments on $[0,1]$ and monotonically decreasing in its third argument. For instance, the derivative with respect to $y$ is $M+2z-2(1-\epsilon)x$ which is non-negative for $x, z\in[0,1]$ since $M\geq 2$. Therefore,  we have
\begin{align*}
f( x,  y,  z)&=f(\max_{s_1}\E_{(s_1)}[X], \max_{s_2}\E_{(s_2)}[Y], \min_{s_3}\E_{(s_1)}[Z])
\\&=\max_{s_1, s_2, s_3} f\big(\E_{(s_1)}[X], \E_{(s_2)}[Y], \E_{(s_3)}[Z]\big).
\end{align*}
To prove the theorem, we thus need to show that
$$g(X, Y, Z):=\max_{s, s_1, s_2, s_3} \bigg(f\big(\E_{(s_1)}[X], \E_{(s_2)}[Y], \E_{(s_3)}[Z]\big)-\E_{(s)}[f(X, Y, Z)]\bigg)\geq 0.$$

Let $X=X'+X''$ where $X'\in \mathcal L_{A}$ and $X''\in \mathcal L_{A}^\perp$. Therefore using~\eqref{eqn:am25} we have
\begin{align}
\mathbb E_{(s)}[X']=\mathbb E_{(s_1)}[X']\qquad \forall s,s_1.\label{eqn;am27}\end{align}
Similarly let $Y=Y'+Y''$ where $Y'\in \mathcal L_{B}$ and $Y''\in \mathcal L_{B}^\perp$. Assume without loss of generality that
\begin{align}
\|X''\|_*\geq \|Y''\|_*.
\label{eqn:am21}
\end{align}
 We now compute
\begin{align}
g(X, Y, Z)&\geq\max_{s, s_1}\, f\big(\E_{(s_1)}[X], \E_{(s)}[Y], \E_{(s)}[Z]\big)-\E_{(s)}[f(X, Y, Z)]\nonumber\\
&=\max_{s, s_1}\, \bigg(M(\E_{(s_1)}[X]+\E_{(s)}[Y])-2(M+\epsilon)\E_{(s)}[Z]+2\E_{(s_1)}[X]\E_{(s)}[Y] \nonumber
\\&\qquad\qquad\qquad -(1-\epsilon)(\E_{(s_1)}[X]^2+\E_{(s)}[Y]^2)\nonumber
\\&\qquad\qquad -\E_{(s)}\big[M(X+ Y)-2(M+\epsilon)Z+2X Y -(1-\epsilon)(X^2+ Y^2)\big]\bigg)\nonumber
\\
&= \max_{s, s_1}\,  M\Big(\E_{(s_1)}[X] - \E_{(s)}[X]\Big) + 2\Big( \E_{(s_1)}[X]\E_{(s)}[Y] - \E_{(s)}[XY]  \Big)\nonumber\\
&\qquad\qquad -(1-\epsilon) \Big( \E_{(s_1)}[X]^2-  \E_{(s)}[X^2]  + \E_{(s)}[Y]^2-  \E_{(s)}[Y^2]\Big)\nonumber\\
&= \max_{s, s_1}\,  M\Big(\E_{(s_1)}[X''] - \E_{(s)}[X'']\Big) + 2\Big( \E_{(s_1)}[X]\E_{(s)}[Y] - \E_{(s)}[XY]  \Big)\nonumber\\
&\qquad\qquad -(1-\epsilon) \Big( \E_{(s_1)}[X]^2-  \E_{(s)}[X^2]  + \E_{(s)}[Y]^2-  \E_{(s)}[Y^2]\Big),\label{eqn:am29}
\end{align}
where in~\eqref{eqn:am29} we use~\eqref{eqn;am27} and the fact that $X=X'+X''$.
By Lemma~\ref{lem:Delta} there are $s, s_1$ such that
\begin{align}\E_{(s_1)}[X''] - \E_{(s)}[X'']\geq \Delta\|X''\|_*.
\label{eqn:anb1}\end{align}
From now on we fix $s, s_1$ to be the ones that achieve the above inequality. By this choice we obtain a lower bound on the first term of~\eqref{eqn:am29}.
\begin{align*}
g(X, Y, Z)& \geq  M\|X''\|_* + 2\Big( \E_{(s_1)}[X]\E_{(s)}[Y] - \E_{(s)}[XY]  \Big)\nonumber\\
&\qquad -(1-\epsilon) \Big( \E_{(s_1)}[X]^2-  \E_{(s)}[X^2]  + \E_{(s)}[Y]^2-  \E_{(s)}[Y^2]\Big).
\end{align*}
To bound the second term of~\eqref{eqn:am29}, we use $X=X'+X''$ and $Y=Y'+Y''$ to write
\begin{align}
\E_{(s_1)}[X]\E_{(s)}[Y]&= \E_{(s_1)}[X']\E_{(s)}[Y'] +  \E_{(s_1)}[X']\E_{(s)}[Y''] +  \E_{(s_1)}[X'']\E_{(s)}[Y'] +  \E_{(s_1)}[X'']\E_{(s)}[Y'']\nonumber\\
&\geq \E_{(s_1)}[X']\E_{(s)}[Y'] -x'_{\max}y''_{\max} -  x''_{\max}y'_{\max} - x''_{\max}y''_{\max},\label{eqn:ab2}
\end{align}
where $x'_{\max}=\max_{a} |X'(a)|$, and $x''_{\max}, y'_{\max}$ and $y''_{\max}$ are defined similarly. Now note that
$\|X\|^2_*=\|X'\|^2_* + \|X''\|^2_*$. Moreover, $\|X\|^2_*\leq 1$ since $X(a)\in [0,1]$ for all $a$. This implies that $x'_{\max}, x''_{\max}\leq  |\sqrt{\mathcal A}|$, and similarly $y'_{\max}, y''_{\max}\leq \sqrt{|\mathcal B|}$. We also have
$$\max \Big\{\frac{1}{\sqrt{|\mathcal A|}}x''_{\max}, \frac{1}{\sqrt{|\mathcal B|}}y''_{\max}\Big\}\leq \max\{\|X''\|_*, \|Y''\|_*\}= \|X''\|_*,$$
where here we use~\eqref{eqn:am21}.
We can then use these inequalities in~\eqref{eqn:ab2} to obtain
\begin{align}
\E_{(s_1)}[X]\E_{(s)}[Y]&\geq \E_{(s_1)}[X']\E_{(s)}[Y'] -3\sqrt{|\mathcal A||\mathcal B|}\cdot\|X''\|_*.
\label{eqn:anb2}
\end{align}
By the same analysis on $\E_{(s)}[XY]$ we get
\begin{align*}
\E_{(s)}[XY]&\leq \E_{(s)}[X'Y']+x'_{\max}y''_{\max} +  x''_{\max}y'_{\max} + x''_{\max}y''_{\max}\\
&\leq  \E_{(s)}[X'Y']+3\sqrt{|\mathcal A||\mathcal B|}\cdot \|X''\|_*.
\end{align*}
As a result,
$$ \E_{(s_1)}[X]\E_{(s)}[Y] - \E_{(s)}[XY] \geq  \E_{(s_1)}[X']\E_{(s)}[Y'] - \E_{(s)}[X'Y'] -6\sqrt{|\mathcal A||\mathcal B|}\cdot\|X''\|_*.$$
Applying the same lines of inequalities for the other terms we obtain
\begin{align}\E_{(s_1)}[X]^2-  \E_{(s)}[X^2]\leq \E_{(s_1)}[X']^2-  \E_{(s)}[X'^2]+6|\mathcal A|\cdot\|X''\|_*,
\label{eqn:anb3}
\end{align}
and
\begin{align}
\E_{(s)}[Y]^2-  \E_{(s)}[Y^2]\leq \E_{(s)}[Y']^2-  \E_{(s)}[Y'^2]+6|\mathcal B|\cdot\|X''\|_*.
\label{eqn:anb4}
\end{align}
Putting equations \eqref{eqn:anb1}, \eqref{eqn:anb2},  \eqref{eqn:anb3} and  \eqref{eqn:anb4} together we obtain
\begin{align*}
g(X, Y, Z)&\geq M\Delta\|X''\|_* + 2\Big(   \E_{(s_1)}[X']\E_{(s)}[Y'] - \E_{(s)}[X'Y'] -6|\mathcal A||\mathcal B|\cdot\|X''\|_* \Big) \\
&\qquad-(1-\epsilon)\Big(  \big( \E_{(s_1)}[X']^2-  \E_{(s)}[X'^2]\big) + \big(  \E_{(s)}[Y']^2-  \E_{(s)}[Y'^2]\big) +12|\mathcal A||\mathcal B|\cdot\|X''\|_*\Big)\\
&\geq (M\Delta-24|\mathcal A||\mathcal B|)\|X''\|_* + 2\Big(   \E_{(s_1)}[X']\E_{(s)}[Y'] - \E_{(s)}[X'Y']  \Big) \\
&\qquad-(1-\epsilon)\Big(   \E_{(s_1)}[X']^2-  \E_{(s)}[X'^2] +   \E_{(s)}[Y']^2-  \E_{(s)}[Y'^2] \Big)\\
& \geq 2\Big(   \E_{(s)}[X']\E_{(s)}[Y'] - \E_{(s)}[X'Y']  \Big)
\\&\qquad-(1-\epsilon)\Big(   \E_{(s)}[X']^2-  \E_{(s)}[X'^2] +  \E_{(s)}[Y']^2-  \E_{(s)}[Y'^2] \Big),
\end{align*}
where in the last line we use~\eqref{eqn;am27} and the fact that $M\Delta-24|\mathcal A||\mathcal B|\geq 0$.

Let
\begin{align}h(X', Y')&= 2\Big(   \E_{(s)}[X']\E_{(s)}[Y'] - \E_{(s)}[X'Y']  \Big) \nonumber
\\&\qquad-(1-\epsilon)\Big(   \E_{(s)}[X']^2-  \E_{(s)}[X'^2] +  \E_{(s)}[Y']^2-  \E_{(s)}[Y'^2] \Big),\label{eqn:hvw}
\end{align}
Then it suffices to show that $h(X', Y')\geq 0$ for every $X'\in \mathcal L_A$ and $Y'\in \mathcal L_B$. By a simple algebra for every $r, t\in \mathbb R$ we verify that
$$h(X'+r\mathbf 1_{A}, Y'+t\mathbf 1_{B})=h(X', Y').$$
This means that with no loss of generality we may assume that $X'\in\mathcal L'_{A}=\mathcal L_A\cap (\mathbf 1_A)^\perp$ and  $Y'\in \mathcal L'_{B}=\mathcal L_B\cap (\mathbf 1_A)^\perp$.

Let $U\in   \mathcal K_A$ and $U'\in \mathcal K_A^{\perp_s}$ such that $X'=U+U'$. Similarly let $Y'=V+V'$ where $V\in  \mathcal K_B$ and $V'\in \mathcal K_B^{\perp_s}$.
Since $U\in  \mathcal K_A$, its values can be denoted by $u_{c}$. We similarly denote the values of $V$ by $v_c$.
Therefore, we have
\begin{align}\label{eq:c-zero-mean}
\E_{(cs)}[U']=\E_{(cs)}[V']=0 \qquad\qquad \forall c.
\end{align}
Thus by the characterization of $\rho=\rho(A, B|CS)$ given in~\eqref{eq:rho-conditional-zero-mean} we have
\begin{align}\label{eq:rho-mean}
\E_{(s)}[U'V']\leq \rho\sqrt{\E_{(s)}[U'^2]\E_{(s)}[V'^2]}\leq \frac{\rho}{2}\Big(\E_{(s)}[U'^2] + \E_{(s)}[V'^2]\Big).
\end{align}
Further~\eqref{eq:c-zero-mean} implies that $\E_{(s)}[U']=\E_{(s)}[V']=0$ and then
\begin{align}\E_{(s)}[X']=\E_{(s)}[U], \qquad \E_{(s)}[Y']=\E_{(s)}[V].\label{eqn;amn1av}
\end{align}
 Moreover using equation \eqref{eq:c-zero-mean} we find that
\begin{align}\E_{(s)}[X'Y']&=\E_{(s)}[UV] + \E_{(s)}[UV'] + \E_{(s)}[U'V] +  \E_{(s)}[U'V']\nonumber
\\&=\E_{(s)}[UV] + \sum_{c}p_s(c)u_c\E_{(cs)}[V'] + \sum_{c}p_s(c)v_c\E_{(cs)}[U'] +  \E_{(s)}[U'V']\nonumber
\\&=\E_{(s)}[UV] + \E_{(s)}[U'V'].\label{eqn;amn2av}
\end{align}
A similar argument shows that
\begin{align}\E_{(s)}[X'^2]&= \E_{(s)}[U^2]+ \E_{(s)}[U'^2],\label{eqn;amn3av}
\\
\E_{(s)}[Y'^2]&= \E_{(s)}[V^2]+ \E_{(s)}[V'^2].\label{eqn;amn4av}
\end{align}

Using equations \eqref{eqn;amn1av}-\eqref{eqn;amn4av}, we compute a lower bound for $h(X', Y')$.
\begin{align*}
h(X', Y')& \geq 2\Big(   \E_{(s)}[U]\E_{(s)}[V] -\E_{(s)}[UV] - \E_{(s)}[U'V']\Big) \\
&\qquad  -(1-\epsilon)\Big(  \E_{(s)}[U]^2 - \E_{(s)}[U^2]-\E_{(s)}[U'^2]   +\E_{(s)}[V]^2 - \E_{(s)}[V^2]-\E_{(s)}[V'^2] \Big).
\end{align*}
Using~\eqref{eq:rho-mean} we continue
\begin{align*}
h(X', Y')& \geq \Big(  2\E_{(s)}[U]\E_{(s)}[V] -2\E_{(s)}[UV] + \E_{(s)}[U^2]- \E_{(s)}[U]^2 + \E_{(s)}[V^2]- \E_{(s)}[V]^2    \Big) \\
&\qquad + \epsilon\Big(  \E_{(s)}[U]^2- \E_{(s)}[U^2] + \E_{(s)}[V]^2- \E_{(s)}[V^2]       \Big) -2\E_{(s)}[U'V'] 
\\ &\qquad+ (1-\epsilon)\Big(  \E_{(s)}[U'^2] +  \E_{(s)}[V'^2]\Big)\\
 &\geq \Big(  2\E_{(s)}[U]\E_{(s)}[V] -2\E_{(s)}[UV] + \E_{(s)}[U^2]- \E_{(s)}[U]^2 + \E_{(s)}[V^2]- \E_{(s)}[V]^2    \Big) \\
&\qquad + \epsilon\Big(  \E_{(s)}[U]^2- \E_{(s)}[U^2] + \E_{(s)}[V]^2- \E_{(s)}[V^2]       \Big) + (1-\epsilon-\rho)\Big(  \E_{(s)}[U'^2] +  \E_{(s)}[V'^2]\Big)\\
%&=  -2\E_{(s)}[(U-\E_{(s)}[U])(V-\E_{(s)}[V])] + \E_{s}[(U-\E_{(s)}[U])^2] + \E_{s}[(V-\E_{(s)}[V])^2]\\
%&\qquad + \epsilon\Big(  \E_{(s)}[U^2]- \E_{(s)}[U]^2 + \E_{(s)}[V^2]- \E_{(s)}[V]^2       \Big) + (1-\epsilon-\rho)\Big(  \E_{(s)}[U'^2] +  \E[V'^2]\Big)\\
&=  \E_{(s)}\Big[\big((U-\E_{(s)}[U])-(V-\E_{(s)}[V])\big)^2\Big]
 \\
&\qquad + \epsilon\Big(  \E_{(s)}[U]^2- \E_{(s)}[U^2] + \E_{(s)}[V]^2- \E_{(s)}[V^2]       \Big)+ (1-\epsilon-\rho)\Big(  \E_{(s)}[U'^2] +  \E_{(s)}[V'^2]\Big)\\
& \geq \epsilon\Big(  \E_{(s)}[U]^2- \E_{(s)}[U^2] + \E_{(s)}[V]^2- \E_{(s)}[V^2]       \Big) + (1-\epsilon-\rho)\Big(  \E_{(s)}[U'^2] +  \E_{(s)}[V'^2]\Big)\\
&\geq  -\epsilon\Big( \E_{(s)}[U^2] +\E_{(s)}[V^2]       \Big) + (1-\epsilon-\rho)\Big(  \E_{(s)}[U'^2] +  \E_{(s)}[V'^2]\Big)\\
&=-\epsilon\big(   \|U\|_{(s)}^2 +\|V\|^2_{(s)}  \big) + (1-\epsilon-\rho)\big(\|U'\|^2_{(s)}+\|V'\|_{(s)}^2\big).
\end{align*}
Now using Lemma~\ref{lem:Delta-prime} we have
$$\|U'\|_{(s)}\geq \Delta'\|U\|_{(s)}, \qquad \qquad \|V'\|_{(s)}\geq \Delta'\|V\|_{(s)}.$$
Hence,
\begin{align*}
h(X', Y')
&\geq-\epsilon\big(   \|U\|_{(s)}^2 +\|V\|^2_{(s)}  \big) + (1-\epsilon-\rho)\Delta'\big(\|U\|^2_{(s)}+\|V\|_{(s)}^2\big)\\
&=(\Delta' (1-\rho)-(1+\Delta')\epsilon)\big(\|U'\|^2_s+\|V'\|_{(s)}^2\big)\\
&\geq 0.
\end{align*}
These inequalities hold since $\epsilon\leq \Delta' (1-\rho)/(1+\Delta')$.

\end{document}